\documentclass[sigconf, nonacm]{acmart}

\usepackage{tikz}
\usepackage{amsmath}

\usepackage{filecontents}

\usepackage[dvipsnames]{xcolor}
\usepackage{xspace}
\usepackage{calligra}
\usepackage{graphicx}
\usepackage{amssymb}
\usepackage{tcolorbox}
\usepackage{multicol}
\usepackage{algorithmic}
\usepackage{setspace}
\usepackage{enumitem}
\usepackage{amsthm}
\usepackage{array}
\usepackage{multirow}
\usepackage{bm}
\usepackage{here}
\usepackage{xcolor}
\usepackage{booktabs}
\usepackage[table]{xcolor}
\usepackage{array}
\usepackage{diagbox}
\usepackage{amsmath}
\usepackage{graphicx}
\usepackage{subfigure}
\usepackage{booktabs}
\usepackage{balance}
\usepackage{multicol}
\usepackage{ulem}

\newcolumntype{B}{!{\vrule width 0.7pt}}

\newtheoremstyle{italicdefinition}
  {\topsep}
  {\topsep}
  {\itshape}
  {}
  {\bfseries}
  {.}
  { }
  {}

\theoremstyle{italicdefinition}
\newtheorem{definition}{\textsc{Definition}}
\theoremstyle{theorem}
\newtheorem{theorem}{\textsc{Theorem}}
\usepackage{adjustbox}
\usepackage{tcolorbox}
\usepackage{float}
\usepackage{caption}
\tcbuselibrary{skins, breakable, theorems}
\usepackage{hhline}
\usepackage{hyperref}
\usepackage{fancybox}

\newcommand{\fname}{\texttt{Bifrost}\xspace}
\newcommand{\baseline}{\texttt{iPrivJoin}\xspace}
\newcommand{\suda}{\texttt{Suda}\xspace}
\newcommand{\CPSI}{\texttt{CPSI}\xspace}

\newcommand{\MPSPDZ}{\texttt{MP-SPDZ}\xspace}

\newcommand{\share}[1]{\langle {#1} \rangle}

\newcommand{\IDb}{\textit{ID}^{b}}
\newcommand{\IDbb}{\pi^b_2(\textit{ID}^{b})}

\newcommand{\IDa}{\textit{ID}^{a}}
\newcommand{\IDaa}{\pi^a_1(\textit{ID}^{a})}
\newcommand{\IDaab}{\pi_1(\textit{ID}^{a})}

\newcommand{\MII}{\mathsf{MIPairs}}
\newcommand{\SMIG}{SMIG\xspace}

\newtcbtheorem[auto counter]{protocol}{\textbf{Protocol}}
  {enhanced, float,
   colback = white,
   attach boxed title to top left = {yshift = -2.8mm, xshift = 3mm},
   left = 1mm,
   right = 1mm,
   top = 2mm,
   bottom = 0.5mm,
   fonttitle = \sffamily\bfseries,
   width = \linewidth}{pro}

\usepackage{framed}
\definecolor{shadecolor}{RGB}{240,240,240}

\definecolor{GreenB}{RGB}{69,133,72}

\newcommand{\rf}[1]{\textcolor{black}{#1}}

\newcommand{\rall}[1]{\textcolor{black}{#1}}
\newcommand{\rone}[1]{\textcolor{black}{#1}}
\newcommand{\rtwo}[1]{\textcolor{black}{#1}}
\newcommand{\rthree}[1]{\textcolor{black}{#1}}

\newcommand\vldbdoi{XX.XX/XXX.XX}
\newcommand\vldbpages{XXX-XXX}
\newcommand\vldbvolume{19}
\newcommand\vldbissue{6}
\newcommand\vldbyear{2026}
\newcommand\vldbauthors{\authors}
\newcommand\vldbtitle{\shorttitle}
\newcommand\vldbpagestyle{empty}

\begin{document}

\title{\fname: A Much Simpler Secure Two-Party Data Join Protocol for Secure Data Analytics}

\author{Shuyu Chen}
\affiliation{
  \institution{Fudan University}
}
\email{23110240005@m.fudan.edu.cn}

\author{Mingxun Zhou}
\affiliation{
  \institution{The Hong Kong University of Science and Technology}
}
\email{mingxunz@ust.hk}

\author{Haoyu Niu}
\affiliation{
  \institution{Fudan University}
}
\email{23212010019@m.fudan.edu.cn}

\author{Guopeng Lin}
\affiliation{
  \institution{Fudan University}
}
\email{17302010022@fudan.edu.cn}

\author{Weili Han}
\affiliation{
  \institution{Fudan University}
}
\email{wlhan@fudan.edu.cn}

\begin{abstract}
Secure data join enables two parties with vertically distributed data to securely compute the joined table, allowing the parties to perform downstream Secure multi-party computation-based Data Analytics (SDA), such as analyzing statistical information or training machine learning models, based on the joined table.
While Circuit-based Private Set Intersection (\CPSI) can be used for secure data join, it inherently introduces redundant dummy rows in the joined table on top of the actual matched rows, which results in high overhead in the downstream SDA tasks.
\baseline addresses this issue but introduces significant communication overhead in the redundancy removal process, as it relies on the cryptographic primitive Oblivious Programmable Pseudorandom Function (OPPRF) and multiple rounds of oblivious shuffles.

In this paper, we propose a much simpler secure data join protocol, \fname, which outputs (the secret shares of) a redundancy-free joined table.
The highlight of \fname lies in its simplicity: it builds upon two conceptually simple building blocks, an ECDH-PSI protocol and a two-party oblivious shuffle protocol. The lightweight protocol design allows \fname to avoid the need for OPPRF. We also proposed a simple optimization named \textit{dual mapping} that reduces the rounds of oblivious shuffle needed from two to one.
Experiments on various datasets up to 100 GB show that \fname achieves $2.54 \sim 22.32\times$ speedup and reduces the communication by $84.15\% \sim 88.97\%$ compared to the state-of-the-art redundancy-free secure data join protocol \baseline. In addition, the communication size of \fname is nearly equal to the size of the input data.
In the experiments involving the entire two-step SDA pipeline (secure join and secure analytics), the redundancy-free property of \fname not only avoids the catastrophic error rate blowup in the downstream tasks caused by the dummy padded outputs in the joined table (as introduced in \texttt{CPSI}), but also shows up to $2.80\times$ speed-up in the secure analytics process with up to $73.15\%$ communication reduction.

\end{abstract}

\maketitle

\pagestyle{\vldbpagestyle}
\begingroup\small\noindent\raggedright\textbf{PVLDB Reference Format:}\\
\vldbauthors. \vldbtitle. PVLDB, \vldbvolume(\vldbissue): \vldbpages, \vldbyear.\\
\href{https://doi.org/\vldbdoi}{doi:\vldbdoi}
\endgroup
\begingroup
\renewcommand\thefootnote{}\footnote{\noindent
This work is licensed under the Creative Commons BY-NC-ND 4.0 International License. Visit \url{https://creativecommons.org/licenses/by-nc-nd/4.0/} to view a copy of this license. For any use beyond those covered by this license, obtain permission by emailing \href{mailto:info@vldb.org}{info@vldb.org}. Copyright is held by the owner/author(s). Publication rights licensed to the VLDB Endowment. \\
\raggedright Proceedings of the VLDB Endowment, Vol. \vldbvolume, No. \vldbissue\
ISSN 2150-8097. \\
\href{https://doi.org/\vldbdoi}{doi:\vldbdoi} \\
}\addtocounter{footnote}{-1}\endgroup

\ifdefempty{\vldbavailabilityurl}{}{
\vspace{.3cm}
\begingroup\small\noindent\raggedright\textbf{PVLDB Artifact Availability:}\\
The source code, data, and/or other artifacts have been made available at:
https://github.com/stellasuc/secdjoin.git
\endgroup
}

\section{Introduction}
\label{sec:intro}

Secure Multi-Party Computation (SMPC) enables multiple parties to jointly compute some functions over their data while keeping their data private.
By leveraging SMPC, parties can perform various Secure Data Analytics (SDA) tasks, including performing statistical analysis on their distributed data or even training a machine learning model, which helps address privacy-related data compliance requirements, such as GDPR~\cite{gdpr} and CCPA~\cite{bonta2022ccpa}.

\textit{Vertically distributed data} is a common setting in real-world SDA applications, appearing in many privacy-critical scenarios, including finance~\cite{secboost21,slg22fl}, e-commerce~\cite{wei2022vertical,chen2020vafl}, and healthcare~\cite{azar2013decision,wu2020privacy}.
That is, the parties' data tables overlap in the indexing IDs but have disjoint feature columns (see Figure~\ref{fig:vppml}). For example, consider a hospital that is cooperating with a genomic research company on training a cancer risk assessment model based on genetic data.
The hospital holds the health record table of its patients, and the genomic research company holds the gene description table of its clients.
Both tables are indexed by the individuals' unique IDs.
They will need to first \textit{align} the data by the IDs and
build a training dataset that contains only the records of those individuals who appear in both original tables.
In the traditional setting with no privacy constraint, the parties can run a classical distributed \textit{join} operation~\cite{14joindistri,rodiger2016flow,17worstcaseparl} on the tables and obtain a joined table that cross-matches the records in the original tables based on IDs and concatenates the feature columns.
They can then proceed to train their models based on the data samples in this joined table.

In our targeted privacy-constrained setting, a \textit{secure join} protocol is required instead of a traditional distributed join algorithm, given that these algorithms inherently leak sensitive information about one party's data table to the other party.
The secure join protocol ensures that the computation process of the join operation reveals only the absolutely necessary information to both parties (e.g., the size of the final joined table), and outputs the joined table in secret-shared form to both parties, so neither party learns the full, privacy-sensitive joined table.

Existing secure data join schemes either introduce significant computation and/or communication costs, or pass down nonideal overhead to the downstream SDA tasks.
On the one hand, Circuit-based Private Set Intersection (\CPSI)~\cite{rindal2021vole,rindal2022blazing} can be used for secure data join, but the protocols will inherently introduce many redundant dummy rows in the joined table (padded to the maximum length) other than the actual matched rows, which could be a huge waste in many practical scenarios.
Liu \textit{et al.}~\cite{liu2023iprivjoin} showed that the redundant rows generated by \CPSI may cause both downstream SDA tasks' time and communication overhead to increase by up to $3\times$,
and proposed another state-of-the-art (SOTA) secure data join method named
\baseline.
However, to remove the redundant rows, \baseline introduces excessive communication overhead in the secure data join process, which stems from its reliance on Oblivious Programmable Pseudorandom Function (OPPRF)~\cite{rindal2021vole} for data encoding and two rounds of oblivious shuffle.
Apart from that, \baseline also requires 11 rounds of communication between the two parties, becoming another disadvantage for its practicality.
Our experiments show that for the table size of 100 GB under WAN setting, $\baseline$ requires 53.87 hours in time cost and 773.95 GB in communication cost.

\begin{figure}[t]
    \centering
    \includegraphics[width=0.47\textwidth]{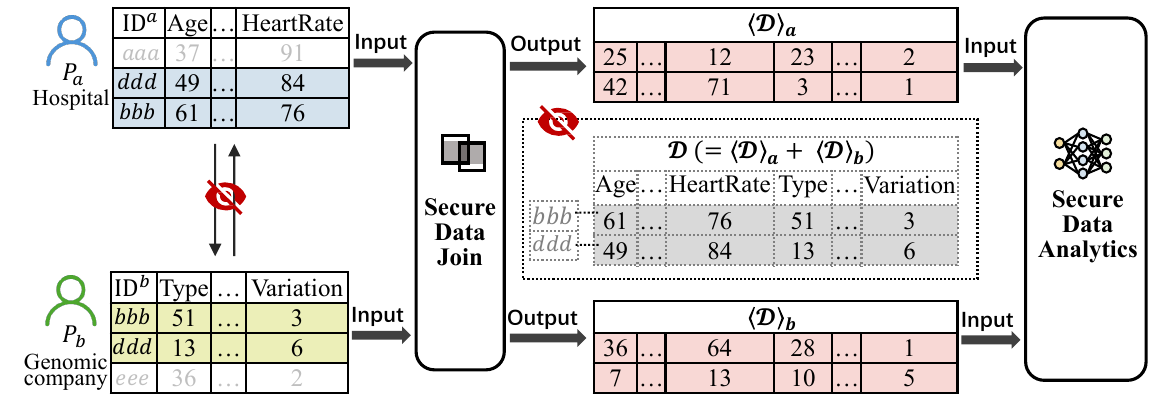}
    \caption{\rtwo{An example of a vertically distributed data setting. \(P_a\) and \(P_b\) each hold a table with three records. Both tables contain the intersection IDs ``$bbb $'' and ``$ddd$''.}}
    \label{fig:vppml}
\end{figure}

As a versatile building block, a secure join protocol can be applied to secure data mining~\cite{agrawal2000privacy,lin2025kona}, SQL operators~\cite{bater2016smcql,tong2022hu}, and data visualizations~\cite{21ppdvisual,21pubvisual} applications, in addition to the aforementioned SDA tasks.
Given the importance of an efficient secure join protocol in these applications, we ask the following question:
\begin{itemize}[leftmargin=6mm]
\item[]
\textit{How to construct a redundancy-free secure data join protocol with high computation efficiency and low communication cost?
}
\end{itemize}

\subsection{Our Contributions}
In this paper, we propose \fname, a conceptually simple secure data join protocol with better performance compared to the baselines.
We summarize our main contributions as follows:
\begin{itemize}[leftmargin=3mm]
    \item We propose a new secure two-party data join protocol named \fname that outputs (the secret shares of) the redundancy-free joined table. The highlight of \fname lies in its simplicity: \fname builds upon two conceptually simple building blocks, an Elliptic-Curve-Diffie-Hellman-based Private Set Intersection (ECDH-PSI) protocol~\cite{meadows1986more,huberman1999enhancing} and a two-party oblivious shuffle protocol~\cite{liu2023iprivjoin}. The lightweight protocol design allows \fname to avoid many performance bottlenecks in \baseline, including the need for Cuckoo hashing and OPPRF. We also propose a simple optimization named \textit{dual mapping} that reduces the rounds of oblivious shuffle (as in \baseline) needed from two to one.

    \item We provide a theoretical analysis of \fname's asymptotic performance and show that it outperforms the baselines in nearly all important metrics, including computation, communication, and round complexities. We also provide the rigorous security proof for \fname to meet the privacy requirement.

    \item We extensively evaluate the effectiveness of \fname in the pipeline of SDA tasks, including secure statistical analysis and secure training.
    First, we show that our redundancy-free design of \fname allows the downstream SDA tasks to achieve high-accuracy outputs, avoiding the catastrophic error rate blowup caused by \texttt{CPSI}-based solution~\cite{rindal2021vole,rindal2022blazing}.
    Moreover, the design helps the downstream SDA tasks achieve up to $2.80\times$ improvement in running time with up to $73.09\%$ less communication compared to pipelines using \texttt{CPSI}.

    \item We empirically evaluate \fname on various datasets up to $100$ GB and compare its performance against the baselines.
    We show that, on real-world datasets, \fname achieves up to $15.2\times$ and up to $10.36\times$ speedups compared to \baseline~\cite{liu2023iprivjoin} and \CPSI~\cite{rindal2022blazing}, respectively.
    In addition, compared to the redundancy-free secure data join baseline \baseline, \fname achieves $2.54 \sim 22.32\times$ faster running time and reduces communication size by $84.15\% \sim 88.97\%$.
    Moreover, the advantages of \fname become more pronounced as the feature dimension increases. For instance, when the feature dimension varies from 100 to 6400, the running time improvement increases from  $9.58\times$ to $21.46\times$. Notably, the communication size of \fname is nearly equal to the size of the input data.

\end{itemize}

\subsection{Design Goal and Technical Overview}
\label{sec:intro_overview}

\vspace{0.5em}
\noindent \textbf{\rtwo{Design Goal.}} \rtwo{
We aim to design a secure two-party data join protocol, \fname.
To clarify this goal, we first define the two-party data join operation~\cite{bater2016smcql,poddar2021senate}. Consider two tables, $\mathsf{Ta}$ and $\mathsf{Tb}$, held by two parties $P_a$ and $P_b$, respectively. Each table consists of identifiers and corresponding feature columns.
A two-party data join operation inputs $\mathsf{Ta}$ and $\mathsf{Tb}$ and outputs a joined table $\mathcal{D}$ by identifying the matched identifiers across both tables and concatenating their respective feature columns.
This join operation can be expressed in SQL as follows:
}

\begin{center}
\begin{small}
\begin{verbatim}
SELECT Ta.* EXCEPT (identifiers), Tb.* EXCEPT (identifiers)
FROM Ta JOIN Tb 
ON Ta.identifiers = Tb.identifiers
\end{verbatim}
\end{small}
\end{center}

\rtwo{
The \textit{secure} two-party data join protocol, \fname, ensures that this join operation is performed securely. Here, ``secure'' means that \fname reveals only the size of the joined table during the computation process of the join operation, and outputs (secret) shares of the joined table to both parties (one share per party).
Note that each joined table share is uniformly random, so it reveals nothing about the joined table value or the matched identifiers. Furthermore, these shares enable the two parties to directly perform downstream secure data analytics on the joined table shares using secure multi-party computation techniques.
}

\begin{figure*}[htbp]
    \centering
    \includegraphics[width=0.98\textwidth]{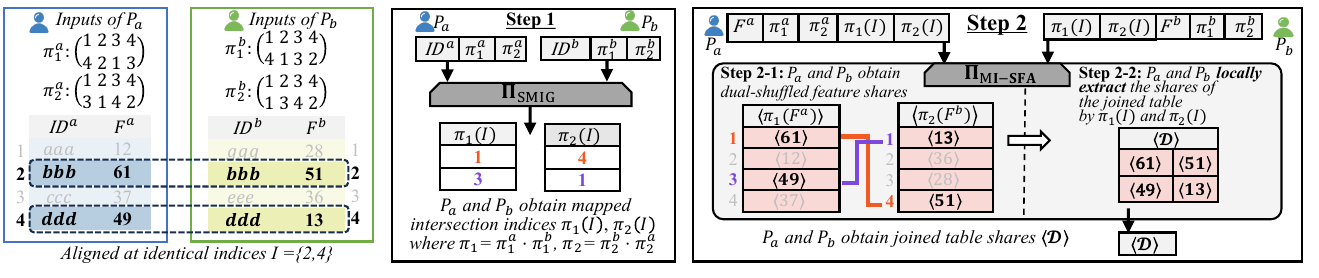}
    \caption{
    Simplified workflow of \fname with \textit{Dual Mapping} optimization. The intersection IDs and their features are highlighted in bold.
    In the first step, both parties obtain mapped intersection indices $\pi_1(I)$ and $\pi_2(I)$.
    In the second step, both parties first obtain dual-shuffled feature shares $\share{\pi_1(F^a)}$ and $\share{\pi_2(F^b)}$ via one round of $\Pi_{\rm O\text{-}Shuffle}$, then both parties locally extract the shares of the joined table $\mathcal{D}$ from $\share{\pi_1(F^a)}$ and $\share{\pi_2(F^b)}$ according to  $\pi_1(I)$ and $\pi_2(I)$.
    }
    \label{fig:overview}
\end{figure*}

\vspace{0.5em}
\noindent \textbf{Simplified Setting Description.} To show how \fname achieves this design goal, we now provide a technical overview in a simplified setting.
There are two tables consisting of identifiers and the corresponding feature columns, $[(\IDa_i, F^a_i)$ $]_{i\in[n]}$
and $[(\IDb_i, F^b_i)]_{i\in[n]}$, held by party $P_a$ and $P_b$, respectively.
Both parties want to run some downstream task on the \textit{joined table}, that is, $\{(F^a_i, F^b_i)\}_{i\in I}$ where the list $I$ includes all matched indices: $I=[i \mid \IDa_i = \IDb_i]$. We focus on designing an interactive protocol between the parties that outputs secret shares of the joined table to both parties (one share per party) without exposing the intersection index set $I$ to either party.

For simplicity, we make the assumption here that the tables are ``aligned'', i.e., if an identifier appears in both tables, the corresponding rows must share the same relative location in the lists.
\rtwo{This ``aligned'' simplicity allows us to focus on the challenge of how to identify $I$ and extract the corresponding feature pairs $\{(F^a_i, F^b_i)\}_{i\in I}$ in secret-shared form without revealing $I$ itself.}
The unaligned case, which is typically the standard definition of the secure join problem, can be accommodated through protocol-specific modifications.

\vspace{0.5em}
\noindent \textbf{Workflow of~\fname.}
The two major steps are:
\begin{enumerate}[leftmargin=*, itemsep=0pt, topsep=0pt, partopsep=0pt]
    \item \textit{Secure mapped intersection generation (SMIG).} In this step, the two parties obliviously find all the matched indices $I$ where $\IDa_i=\IDb_i$ for $i\in I$. For privacy, neither party can learn the indices $I$ directly.
    Instead, our proposed protocol outputs a \textit{mapped} version of the matched indices $I$, denoted as $\pi(I)$, where $\pi:[n]\mapsto[n]$ is a mapping (shuffling) unknown to both parties. We set $\pi$ as the composition of two shuffles $\pi^a,\pi^b$, each only known to party $P_a$ and $P_b$ respectively. This ensures that both parties learn only the size of the intersection (i.e., row count of the joined table) and nothing else.
    To securely instantiate this step, we adapt an Elliptic-Curve-Diffie-Hellman-based Private Set Intersection (ECDH-PSI) protocol~\cite{meadows1986more,huberman1999enhancing} to our setting with several technical changes.

    \item \textit{Secure feature alignment.} The second step is straightforward: the parties jointly run an oblivious shuffling protocol~\cite{liu2023iprivjoin} to obtain the shares of $\pi(F^a)$ and $\pi(F^b)$\footnote{We abuse the notation here that $\pi(F^a)_{\pi(i)}=F^a_{i}$ for all $i\in[n]$.}. Then, both parties can locally extract the shares of the joined table. That is, they collect the shares of $\left\{\left(\pi(F^a)_{i}, \pi(F^b)_{i}\right)\right\}_{i\in \pi(I)}$, which is exactly $\left\{\left(F^a_{i}, F^b_i\right)\right\}_{i\in I}$.
\end{enumerate}

\noindent \textit{Example}: the two parties hold $[(aaa, 12), (bbb, 61), (ccc, 37), (ddd, $ $49)]$ and $[(ggg, 28), (bbb, 51), (eee, 36), (ddd,13)]$, respectively.
The matched indices are $I=[2,4]$, corresponding to the intersection identifers ``$bbb$'' and ``$ddd$''.
Given a shuffle \( \pi: (1 \to 3,\ 2 \to 1,\ 3 \to 4, 4 \to 2) \) unknown to both parties, the secure mapped intersection generation protocol outputs the mapped intersection indices $\pi(I)=\pi([2, 4])=[1, 2]$.
Then, both parties run the oblivious shuffle protocol~\cite{liu2023iprivjoin} to shuffle their local feature columns according to $\pi$ and obtain the shares of the two shuffled feature columns: $[61, 49, 12, 37]$ and $[51, 13, 28, 36]$. Now, given the mapped intersection indices $[1,2]$, they simply collect the first and the second entries in the local shares, which are the shares of the concatenated intersection features: $[(61, 51), (49, 13)]$.

\vspace{0.5em}
\noindent \textbf{Optimization: Dual Mapping.}
A performance bottleneck of the aforementioned workflow is the secure shuffle protocol in the second step. The two-party oblivious shuffle protocol $\Pi_{\rm O\text{-}Shuffle}$~\cite{liu2023iprivjoin} is designed for the following functionality: given the input of a data column (from one party) and a shuffle (from another party), the protocol outputs the shares of the shuffled data column to both parties.
Recall that the shuffle $\pi$ is a composition of $\pi^a$ and $\pi^b$ (first applying $\pi^a$, then applying $\pi^b$). When the parties run the oblivious shuffle protocol to shuffle the feature columns of $F^a$, $P_a$ can first locally shuffle the feature columns according to $\pi^a$. Then, they can directly call $\Pi_{\rm O\text{-}Shuffle}$ given the self-shuffled features $\pi^a(F^a)$ from $P_a$ and the shuffle $\pi^b$ from $P_b$ to obtain the share of $\pi(F^a)$.
However, when they need to shuffle the column of $F^b$, $P_b$ cannot shuffle the column locally according to $\pi^b$ first, because the shuffling operation is not commutative (i.e., $\pi=\pi^a\circ\pi^b$ may not be the same as $\pi^b\circ\pi^a$).
A naive solution could be running the $\Pi_{\rm O\text{-}Shuffle}$ twice, resulting in high communication overhead.

We propose a simple yet effective optimization called \textit{Dual Mapping} to address this issue. That is, we now modify the secure mapped intersection generation (SMIG) protocol to output \textit{two} mapped versions of the intersection indices $I$ with two different mappings, $\pi_1$ and $ \pi_2$. Importantly, we let $\pi_1=\pi^a_1\circ\pi^b_1$ and $\pi_2=\pi^b_2\circ\pi^a_2$
where $\pi^a_1, \pi^a_2$ are only known to $P_a$, and $\pi^b_1, \pi^b_2$ are only known to $P_b$.
Then, in the oblivious shuffle protocol execution later, $P_a$'s feature columns will be shuffled by $\pi_1$, while $P_b$'s feature columns will be shuffled by $\pi_2$.
A simplified workflow of \fname with the \textit{Dual Mapping} optimization is shown in Figure~\ref{fig:overview}.
The advantage is that now both first layers of the shuffles can be done by the corresponding party locally, and we only need to execute the oblivious shuffle protocol once on both sides.
Effectively, this optimization saves $\mathcal{O}(m_b)$ online communication overhead compared to the naive solution, where $m_b$ denotes the feature dimension for $P_b$.

\vspace{0.5em}
\noindent \textbf{Supporting the General Scenario.}
In practice, the tables from both parties are usually not aligned, i.e., the rows with the same identifier may not be in the same relative location in the two tables.
The prior standard solutions to this issue, including \CPSI~\cite{rindal2022blazing} and \baseline~\cite{liu2023iprivjoin}, use a combination of Cuckoo hashing~\cite{pinkas2019efficient} and simple hashing to place the data elements in the same relative location in the hash tables.
This further complicates the protocol design and affects the overall performance.

Instead, our protocol is naturally capable of handling the general scenario without requiring any hashing technique involved, saving another $\mathcal{O}(hm_b\kappa)$ in computation cost and $\mathcal{O}(hm_b\kappa)$ in communication cost, where $\kappa$ is the computational security parameter and $h$ is the number of hashing functions.
More specifically, our protocol simply outputs the mapped intersection index pairs $\mathsf{MIPairs} = \{(\pi_1(i),\pi_2(j)) \mid \IDa_i = \IDb_j\}$ in the first step, and proceeds as before.
We demonstrate that this upgrade can be implemented with a minor modification to our aforementioned ECDH-PSI-based~\cite{meadows1986more} secure mapped intersection generation protocol, incurring no additional overhead.

\vspace{0.5em}
\noindent \textbf{Comparison with the Prior State-of-the-art Scheme.}
Compared to the prior best scheme \baseline, our proposed \fname achieves significantly better efficiency in terms of communication and computation overhead by fundamentally restructuring the workflow.
Specifically, \baseline operates in three steps:
(1) private data encoding: the parties use Cuckoo hashing or simple hashing tables for data allocation and OPPRF~\cite{kolesnikov2017practical} for data encoding, obtaining shares of $(B,F^a,F^b)$ with $B_i=0$ for intersection rows and a random value otherwise; (2) oblivious shuffle: the two parties employ two rounds of an oblivious shuffle protocol on encoded data to obtain shares of $\pi(B,F^a,F^b)$. Here, $\pi$ is the composite of two shuffles $\pi^a$, $\pi^b$, each only known to $P_a$ and $P_b$, respectively. (3) private data trimming: the two parties reconstruct $\pi(B)$ in plaintext and remove redundant rows from shares of $\pi(F^a,F^b)$ according to $\pi(B)$.
The bottlenecks in \baseline are twofold:
Firstly, in the first step,  all data from the two parties must be communicated, particularly through OPPRF, incurring both communication and computational complexities of $\mathcal{O}(hn(\lambda + \log n + m_b\kappa) + nm_a\ell)$. Here $\lambda$ and $\kappa$ are security parameters, $\ell$ is the element bit length, and $h$ is the number of hash functions.
Secondly, in the second step, two rounds of $\Pi_{\rm O\text{-}Shuffle}$ on encoded data are required due to the secret-shared nature of the encoded data, incurring both communication and computation complexities of $2n(m_a\ell +m_b\ell+\kappa)$.

In contrast, \fname builds upon two conceptually simple building blocks, an ECDH-PSI protocol~\cite{meadows1986more,huberman1999enhancing} and the oblivious shuffle protocol~\cite{liu2023iprivjoin}, eliminating Cuckoo hashing and OPPRF. Notably, ECDH-PSI overhead $\mathcal{O}(n\sigma)$ (where $\sigma$ is the ECC key bit-length), depends only on identifiers and not on feature dimensions.
We further propose a simple optimization named \textit{dual mapping} that reduces the need for two rounds of oblivious shuffle (as in \baseline) to one round. As a result, the total communication and computation complexity of \fname are both $\mathcal{O}(n\sigma + nm_a\ell + nm_b\ell)$,
reducing overhead by at least 66.7\% compared to \baseline.
As shown in Table~\ref{tab:commcmp}, \fname exhibits substantially lower online communication overhead than \baseline. Moreover, \fname reduces offline communication at least by half compared to \baseline. A more detailed communication comparison is provided in Appendix~\ref{app:comm_cmp}.

\begin{table}[htbp]
\belowrulesep=0pt
\aboverulesep=0pt
\centering
\caption{Comparison of \fname and \baseline\cite{liu2023iprivjoin} in online communication. \rf{Here, $n$ is the row count of the input data; $m_a$ and $m_b$ are the feature dimensions for $P_a$ and $P_b$, respectively; $m = m_a + m_b$; $\ell$ is the element bit length; $\sigma$ is the ECC key bit length; $h$ is the number of hash functions; $\lambda$ and $\kappa$ are the statistical security parameter and the computational security parameter, respectively.}}
\scalebox{0.88}{
\renewcommand{\arraystretch}{1.05}
\setlength{\tabcolsep}{3pt}
\begin{tabular}{c|c|c|c}
\toprule
\multicolumn{1}{c|}{Protocol} & \multicolumn{2}{c|}{Communication Size (bits)} & Round               \\ \toprule
\multirow{2}{*}{\small \baseline \cite{liu2023iprivjoin}}    & Step (1)                            & Step (2)+(3)      & \multirow{2}{*}{11} \\ \cline{2-3}
                              & $\mathcal{O}(hn(\lambda + \log n + m_b\kappa) + nm_a\ell)$    & $\mathcal{O}(nm\ell + n\kappa)$     &                     \\ \bottomrule
{\small \fname}                          & \multicolumn{2}{c|}{$\mathcal{O}(n\sigma + nm\ell)$}                    & 4
\\ \bottomrule
\end{tabular}
}
\label{tab:commcmp}
\end{table}

\section{Preliminaries}
\label{sec:preliminaries}

\par\vspace{0.5em}
\noindent \textbf{Notations.}
We use the notation $\{x_1,\dots,x_t\}$ to denote an unordered set \rf{and} the notation $[x_1,\dots,x_t]$ to denote an ordered list.

\par\vspace{0.5em}
\noindent \textbf{ECDH-PSI.}
\rone{Private Set Intersection (PSI) \rf{protocols} constructed using Elliptic Curve Cryptography (ECC)~\cite{koblitz1987elliptic,miller1985use} are commonly referred to as ECDH-PSI~\cite{meadows1986more,huberman1999enhancing}.
ECC is typically preferred in PSI because it offers the same security level with significantly smaller key sizes compared to other cryptographic algorithms with similar functionality, such as RSA.
Technically, ECC relies on the algebraic properties of elliptic curves over finite fields. }
An elliptic curve $E$ over $\mathbb{Z}_q$ is defined by the equation $\rtwo{y^2 = x^3 + \gamma_1 x + \gamma_2} \pmod{q},$ where $\rtwo{\gamma_1, \gamma_2} \in \mathbb{F}_q$ satisfy the non-singularity condition $\rtwo{4\gamma_1^3 + 27\gamma_2^2} \not\equiv 0 \pmod{q}$.
Scalar multiplication on this curve, denoted as $k\rtwo{Q}$, represents the repeated addition of a point \rtwo{$Q$} to itself $k$ times.
The security of ECC is based on the hardness of the elliptic curve discrete logarithm problem (ECDLP)~\cite{hankerson2004guide}: given two points \rtwo{$Q$} and $Q_2 = k\rtwo{Q}$, it is computationally infeasible to recover $k$.

In a typical ECDH-PSI protocol, two parties $P_a$ and $P_b$ hold data $X=[x_i]_{i \in [n]}$ and $Y=[y_i]_{i \in [n]}$, respectively. They first map their input data to points on an elliptic curve using a cryptographic hash function $H(\cdot)$. Then, they compute the intersection set as follows:
\begin{enumerate}[label=(\arabic*),leftmargin=*,itemsep=0pt,topsep=0pt]
    \item  $P_a$ and $P_b$ each generates a private scalar key, denoted $\alpha$ and $\beta$, respectively.
    \item $P_b$ computes \( \beta Y = [\beta H(y)]_{y \in Y} \)  and sends $\beta Y$ to $P_a$.
    \item $P_a$ shuffles its data to obtain $\pi(X)$ and computes \( \alpha \pi(X) = [ \alpha H(x)]_{x \in \pi(X)} \). $P_a$ receives $\beta Y$ and computes \( \alpha\beta Y = [ \alpha y]_{y \in \beta Y} \). $P_a$ sends $\alpha \pi(X)$ and $\alpha\beta Y$ to $P_b$.
    \item Upon receiving $\alpha \pi(X)$ and $\alpha\beta Y$, $P_b$ computes $\beta \alpha \pi(X) = [ \beta x$ $]_{x \in \alpha \pi(X)}$. Since scalar multiplication on elliptic curves is commutative (i.e., $\alpha\beta = \beta \alpha$), $P_b$ finally obtains the intersection set by evaluating $\beta\alpha \pi(X) \cap \alpha\beta Y$.
    \rtwo{$P_b$ can reveal their intersection elements to $P_a$ if necessary.}
\end{enumerate}

\par\vspace{0.5em}
\noindent \textbf{Additive Secret Sharing.}
Additive secret sharing is a key technology in SMPC~\cite{yao1986generate,evans2018pragmatic}. In the two-party setting of additive secret sharing, an \(\ell\)-bit value \(x\) is split into two secret-shared values (i.e., shares), denoted $\share{x}_a$ and $\share{x}_b$, where \( \share{x}_* \in \mathbb{Z}_{2^\ell} \) and each party \(P_*\) holds one share \(\share{x}_*\) for $* \in \{a,b\}$. The original value \(x\) can be revealed by summing all shares:
$ \share{x}_a + \share{x}_b \equiv x \pmod{2^\ell}$.

\par\vspace{0.5em}
\noindent \textbf{Oblivious Shuffle.}
The oblivious shuffle functionality $\mathcal{F}_{\textit{O-Shuffle}}$ inputs a random permutation $\pi$ from $P_b$ and inputs a matrix $X$ from $P_a$. It outputs a secret-shared and shuffled matrix $\share{\pi(X)}$ while ensuring that \rtwo{the permutation $\pi$ is unknown to $P_a$ and $X$ is unknown to $P_b$}.
In this paper, we employ the oblivious shuffle protocol $\Pi_{\rm O\text{-}Shuffle}$ (Protocol~\ref{pro:oshuffle}).
In the online phase of Protocol~\ref{pro:oshuffle}, the communication size and round are $n\cdot m$ and 1, respectively, where $n\cdot m$ is the size of matrix $X$.
Additionally, for the offline phase of $\Pi_{\rm O\text{-}Shuffle}$, we employ \rf{the} protocol in Algorithm 2 in \cite{liu2023iprivjoin}.
\begin{small}
\begin{protocol}{$\Pi_{\rm{O\text{-}Shuffle}}$}{oshuffle}
\textbf{Parameters:} The size of matrix $n \times m$; The bit length of element $\ell$.
\\
\textbf{Inputs:} $P_a$ holds a matrix $X \in {\mathbb{Z}^{n \times m}_{2^\ell}}$.
\\
\textbf{Outputs:} Each party obtains $\share{\pi(X)}$, where $\pi$ is only known to $P_b$.
\\
\textbf{Offline:} $P_a$ generates a random matrix $R \in \mathbb{Z}^{n \times m}_{2^\ell}$.
$P_b$ samples a random permutation $\pi: [n] \mapsto [n]$. Then, two parties invoke an instance of $\mathcal{F}_{\textit{O-Shuffle}}$ to obtain $\share{\pi(R)}$.
\\
\textbf{Online:}
\begin{enumerate}[label=\arabic{enumi}.,leftmargin=1.2em,itemsep=0pt, topsep=0pt, partopsep=0pt, parsep=0pt]
\item $P_a$ computes $\hat{X} = X-R$ and sends $\hat{X}$ to $P_b$.
\item $P_a$ sets $\share{\pi(X)}_a = \share{\pi(R)}_a$.
\item $P_b$ computes $\share{\pi(X)}_b =  \pi(X)+\share{\pi(R)}_b$.
\end{enumerate}
\end{protocol}
\end{small}
\section{Problem Description}
\label{sec:overview}

We summarize notations we use in Table~\ref{tab:notations}.

\begin{table}[htbp]
\belowrulesep=0pt
\aboverulesep=0pt
\centering
\caption{Notations used in this paper.}
\scalebox{0.87}{
\begin{tabular}{c|p{7cm}}
\toprule
Symbol & Description
\\
\toprule
$P_a$,$P_b$  & Two parties involved in secure data join.
\\
$n$ & The row count of the input data.
\\
$m_a$,$m_b$ & The feature dimension of $P_a$ and $P_b$, respectively.
\\
$m$ & The total feature dimension ($m = m_a + m_b$).
\\
$c$ & \begin{tabular}[c]{@{}c@{}} \rtwo{The row count of the joined table.}  \end{tabular}
\\
$\IDa$, $\IDb$ & The identifers (IDs) of $P_a$ and $P_b$, respectively.
\\
$F^a$, $F^b$ & The features of $P_a$ and $P_b$, respectively.
\\
$\mathcal{D}$ & The redundancy-free joined table.
\\
\rtwo{$|\IDa \cap \IDb|$} & \rtwo{The size of the intersection identifiers.}
\\
$\mathbb{Z}_{2^\ell}$ & \begin{tabular}[c]{@{}c@{}} The ring size used in our paper.\end{tabular}
\\
\rtwo{$\mathbb{Z}_{2^\ell}^{c\times m}$} & \begin{tabular}[c]{@{}c@{}}
\rtwo{
All $c \times m$ matrices whose elements are from the ring $\mathbb{Z}_{2^{\ell}}$.
}
\end{tabular}
\\
$\share{x}_*$ & \begin{tabular}[c]{@{}c@{}} The share (i.e., secret-shared value) of $x$ for $P_*$.\end{tabular}
\\
$X_i$ & \begin{tabular}[c]{@{}c@{}} The $i$-th element of $X$. \end{tabular}
\\
$X_{i,j}$ & \begin{tabular}[c]{@{}c@{}} The $j$-th element of $i$-th row in $X$. \end{tabular}
\\
$\{x_1,\ldots,x_t\}$ & \begin{tabular}[c]{@{}c@{}} An unordered set. \end{tabular}
\\
$[x_1,\ldots,x_t]$ & \begin{tabular}[c]{@{}c@{}} An \rf{ordered} list. \end{tabular}
\\
$[x]$ & \begin{tabular}[c]{@{}c@{}} An \rf{ordered} list $[1,2,...,x]$. \end{tabular}
\\
\bottomrule
\end{tabular}
}
\label{tab:notations}
\end{table}

\noindent \textbf{Problem Definition.}
\rtwo{To formally define the problem, we present the functionality $\mathcal{F}_{\textit{2PC-DJoin}}$ of the secure two-party data join. Note that $\mathcal{F}_{\textit{2PC-DJoin}}$ specifies how to realize a secure data join with a third trusted party (TTP), while our proposed \fname realizes it without such a TTP. As shown in Figure~\ref{funct:secdjoin}, $\mathcal{F}_{\textit{2PC-DJoin}}$ inputs \((\IDa , F^a)\) from $P_a$ and inputs \((\IDb , F^b)\) from $P_b$. Here, \(\textit{ID}^*\) is an ordered list of $n$ identifiers (IDs) and \(F^*\) is the corresponding features of size $n \times m_*$ for $* \in \{a,b\}$.
\(\mathcal{F}_{\textit{2PC-DJoin}}\) outputs shares \(\share{\mathcal{D}}_a\) and \(\share{\mathcal{D}}_b\) of the joined table $\mathcal{D}$ to $P_a$ and $P_b$, respectively. Specifically, the joined table $\mathcal{D}$ contains $c = |\IDa \cap \IDb|$ rows, where each row consists of $m_a$ features from $P_a$ and $m_b$ features from $P_b$ for an intersection ID.
The joined table shares satisfy $\share{\mathcal{D}}_a + \share{\mathcal{D}}_b = \mathcal{D}$ and are sampled uniformly at random.
Since each party only obtains a joined table share that is independent of $\mathcal{D}$, \rf{neither party learns} the intersection IDs nor the feature values in $\mathcal{D}$.
Consequently, after executing \(\mathcal{F}_{\textit{2PC-DJoin}}\), neither party learns any extra information beyond what is revealed by the size $c \times m$ of the joined table. }

\begin{small}
\begin{figure}[h]
    \centering
    \scalebox{0.98}{
    \fbox{
        \parbox{0.46\textwidth}{
        {
        \centering
        \textbf{\underline{Functionality $\mathcal{F}_{\rm \textit{2PC-DJoin}}$}}\par}

        \textbf{Parameters:}
        The input data row count $n$.
        The bit length of element $\ell$.
        $P_a$'s and $P_b$'s feature dimensions $m_a$ and $m_b$.
        $m = m_a + m_b$.

        \textbf{Input:} $P_a$ inputs its data $(\IDa, F^a)$. $P_b$ inputs its data $(\IDb, F^b)$.

        \textbf{Functionality:}
        \begin{enumerate}[label=\arabic{enumi}.,leftmargin=1.2em,itemsep=0pt, topsep=0pt, partopsep=0pt, parsep=0pt]
            \item Let intersection identifiers $\textit{ID}^{a\cap b} = \textit{ID}^{\,a} \cap \textit{ID}^{\,b}$. Define the joined table row count as $c = |\textit{ID}^{a\cap b}|$.
            \item Let joined table $\mathcal{D} = \{[F^a_{j_a,1},\ldots,F^a_{j_a,m_a},F^b_{j_b,1},\ldots, F^b_{j_b,m_b}]\}_{j \in [c]}$, where $j_a$ and $j_b$ are indices such that $\textit{ID}^a_{j_a} = \textit{ID}^b_{j_b} = \textit{ID}^{a \cap b}_j$.
            \item Sample two random tables $\share{\mathcal{D}}_a \in {\mathbb{Z}_{2^\ell}^{c\times m}}$ and $\share{\mathcal{D}}_b \in {\mathbb{Z}_{2^\ell}^{c\times m}}$ such that $\share{\mathcal{D}}_a + \share{\mathcal{D}}_b = \mathcal{D}$.
            \item Return $\share{\mathcal{D}}_a$ to $P_a$ and $\share{\mathcal{D}}_b$ to $P_b$.
        \end{enumerate}
        }
    }
    }
    \caption{Ideal functionality of Secure Two-Party Data Join.}
    \label{funct:secdjoin}
\end{figure}
\end{small}

\par\vspace{0.5em}
\noindent \textbf{Security Model.}
\rall{Following prior secure data join schemes~\cite{liu2023iprivjoin,songsuda}, in this paper, we consider semi-honest probabilistic polynomial time (PPT) adversaries}. A semi-honest PPT adversary $\mathcal{A}$ can corrupt one of the parties $P_a$ or $P_b$ at the beginning of the protocol and aims to learn extra information from the protocol execution while still correctly executing the protocol.
We follow the standard simulation-based definition of security for secure two-party computation~\cite{canetti2001universally,oded2009foundations}.
We give the formal security definition as follows:
\begin{definition}
[Semi-honest Model]
Let $\text{view}^\Pi_C(x,y)$ be the views (including input, random tape, and all received messages) of the corrupted party $C$ during the execution of the protocol $\Pi$, $x$ be the input of the corrupted party and $y$ be the input of the honest party. Let $out(x,y)$ be the protocol's output of all parties and $\mathcal{F}(x,y)$ be the functionality's output. $\Pi$ is said to securely compute a functionality $\mathcal{F}$ in the semi-honest model if for any Probabilistic Polynomial-Time (PPT) adversary $\mathcal{A}$ there exists a simulator $\text{Sim}_C$ such that for all inputs $x$ and $y$,
\[
\{\text{view}^\Pi_C(x,y), \text{out}(x,y)\} \approx_c \{\text{Sim}_C(x, \mathcal{F}_C(x,y)), \mathcal{F}(x,y)\}.
\]
where $ \approx_c$ represents computational indistinguishability.
\end{definition}

\section{Design}
\label{sec:design}

\fname is instantiated by Protocol~\ref{pro:2pcdjoin}, which realizes the secure two-party data join functionality $\mathcal{F}_{\textit{2PC-DJoin}}$.
For simplicity, we assume both parties hold input data with equal row counts ($n$).
\fname can be directly extended to support parties holding data with different row counts without extra overhead, as detailed in the Appendix~\ref{app:difrowcounts}.

\begin{small}
\begin{protocol}{$\Pi_{\rm 2PC\text{-}DJoin}$}{2pcdjoin}
\textbf{Parameters:} The same as $\mathcal{F}_{2PC\text{-}DJoin}$.

\textbf{Input:}  $P_a$ inputs its data $(\IDa, F^a)$. $P_b$ inputs its data $(\IDb, F^b)$.

\textbf{Output:} $P_a$ and $P_b$ obtain the joined table shares $\share{\mathcal{D}}_a$ and $\share{\mathcal{D}}_b$, respectively.

\begin{enumerate}[label=\arabic{enumi}.,leftmargin=1.2em]

\item \textbf{[Secure Mapped Intersection Generation]}. $P_a$ and $P_b$ execute the secure mapped intersection generation protocol $\Pi_{\rm \SMIG}$ (Protocol~\ref{pro:smig}), where $P_a$ inputs its identifiers $\IDa$, and $P_b$ inputs its identifiers $\IDb$.
 After the execution, both parties obtain the mapped intersection index pairs $\MII$.

\item \textbf{[Secure Feature Alignment]}. $P_a$ and $P_b$ execute the mapped intersection-based secure feature alignment protocol $\Pi_{\rm MI\text{-}SFA}$ (Protocol~\ref{pro:sfa}), where $P_a$ inputs its feature columns $F^a$ and mapped intersection pairs $\MII$, and $P_b$ inputs its feature columns $F^b$ and mapped intersection pairs $\MII$. After the execution, $P_a$ and $P_b$ obtain the joined table shares $\share{\mathcal{D}}_a$ and $\share{\mathcal{D}}_b$, respectively.

\end{enumerate}
\end{protocol}
\end{small}

\par\vspace{0.25em}
\noindent \textbf{Setup, Offline and Online Phase.}
Similar to prior works~\cite{hao2024unbalanced,Kerschbaum2023}, operations in \fname are performed in three distinct phases:
(1) offline phase: the two parties perform input-independent precomputations to generate correlated randomness;
(2) setup phase: the two parties preprocess their input data, respectively;
(3) online phase: the two parties use randomness and the preprocessed data to securely compute the output.

\par\vspace{0.25em}
\noindent \textit{\rtwo{Example:}} \rtwo{Throughout this section we use a simple example where $P_a$ holds $(\IDa,F^a) = [(aaa,12), (bbb,61), (ccc,37), (ddd,49)]$, and $P_b$ holds $(\IDb,F^b) = [(ddd,13), (bbb,51), (ggg,28),(eee,36)]$. In the offline phase, $P_a$ samples two permutations $\pi^a_1:(1 \to 4, 2 \to 2, 3 \to 1, 4 \to 3)$ and $\pi^a_2:(1 \to 3, 2 \to 1, 3 \to 4, 4 \to 2)$, and $P_b$ samples two permutations $\pi^b_1:(1 \to 4, 2 \to 1, 3 \to 3, 4 \to 2)$ and $\pi^b_2:(1 \to 1, 2 \to 3, 3 \to 4, 4 \to 2)$.
}

\subsection{Secure Mapped Intersection Generation}
\label{subsec:smig}

In the secure mapped intersection generation step, the two parties execute our proposed secure mapped intersection generation (\SMIG) protocol $\Pi_{\rm \SMIG}$ (Protocol~\ref{pro:smig}). Below, we introduce the definition and construction of this protocol.

\subsubsection{Definition}
As shown in Figure~\ref{funct:smig}, we present a secure mapped intersection generation functionality $\mathcal{F}_{\textit{\SMIG}}$.
$\mathcal{F}_{\textit{\SMIG}}$ takes as input identifiers $\IDa$ from $P_a$ and identifiers $\IDb$ from $P_b$.
It outputs to both parties only the mapped intersection index pairs, denoted by $\MII$. Specifically, $\MII= [\pi_1(i), \pi_2(j) \mid $ $(i,j) \in \mathsf{Pairs}_\uparrow ]$, where $\mathsf{Pairs}_\uparrow$ consists of all intersection index pairs $(i,j)$ with $\IDa_i=\IDb_j$, ordered by $i$, and $\pi_1,\pi_2$ are permutation (shuffle) functions unknown to both parties.
Based on Theorem~\ref{theorem:c}, after the execution of $\mathcal{F}_{\textit{\SMIG}}$, neither party should learn any extra information beyond what is revealed by the row count $c$ of the joined table.

\begin{small}
\begin{figure}[h]
    \centering
    \fbox{
    \scalebox{0.9}{
        \parbox{0.49\textwidth}{
        {
        \centering \textbf{\underline{Functionality $\mathcal{F}_{\textit{\SMIG}}$}}\par}

        \textbf{Parameters:} The input data row count $n$.

        \textbf{Input:} $P_a$ inputs its identifiers $\IDa$. $P_b$ inputs its identifiers $\IDb$.

        \textbf{Functionality:}
        \begin{enumerate}[label=\arabic{enumi}.,leftmargin=1.2em,itemsep=0pt, topsep=0pt, partopsep=0pt, parsep=0pt]
        \item $P_a$ samples two permutations $\pi^a_1,\pi^a_2: [n] \mapsto [n]$. $P_b$ samples two permutations $\pi^b_1, \pi^b_2: [n] \mapsto [n]$. Let $\pi_1=\pi^a_1\circ\pi^b_1$ and $\pi_2=\pi^b_2\circ\pi^a_2$.
        \item Let $\textsf{Pairs} = [ (i,j) \mid \IDa_i == \IDb_j ]$ and let $\mathsf{Pairs}_\uparrow$ denote $\mathsf{Pairs}$ sorted in ascending order by the first element $i$ of each pair.
        \item Define the mapped intersection index pairs $\mathsf{MIPairs} = [\pi_1(i), \pi_2(j)]_{(i,j) \in \mathsf{Pairs}_\uparrow }$.
        \item Return $\mathsf{MIPairs}$ to both $P_a$ and $P_b$.
        \end{enumerate}
        }
    }
    }
    \setlength{\abovecaptionskip}{0pt}
    \setlength{\belowcaptionskip}{0pt}
    \smallskip
    \caption{Ideal functionality of \SMIG.}
    \label{funct:smig}
\end{figure}
\end{small}

\begin{theorem}
\label{theorem:c}
The functionality $\mathcal{F}_{\textit{\SMIG}}$ in Figure~\ref{funct:smig} reveals no information beyond what is revealed by the joined table row count $c$ to both parties.
\begin{proof}
    The proof is presented in Appendix~\ref{app:proof}.
\end{proof}

\end{theorem}

\subsubsection{Construction}
\label{sec:smig_construction}

As shown in Protocol~\ref{pro:smig}, the secure mapped intersection generation protocol $\Pi_{\rm \SMIG}$ consists of an offline phase, a setup phase, and an online phase. The online phase of $\Pi_{\rm \SMIG}$ consists of an ECC-based ID encrypting step and a mapped intersection generation step. Below, we introduce these steps in detail.

\noindent \textbf{Offline Phase.}
Each party $P_*$ ($* \in \{a,b\}$) samples two distinct random permutations $\pi^*_1,\pi^*_2 \in [n]\mapsto[n]$. Both parties keep their permutations for future use.

\noindent \textbf{Setup Phase.}
Each party first locally shuffles its identifiers (IDs) and then encrypts the shuffled IDs using the ECC key. Specifically,
(a) For $P_a$, it first shuffles its IDs with $\pi^a_1$ to obtain the shuffled IDs $\IDaa$. Next, $P_a$ generates an ECC key $\alpha$ and uses $\alpha$ to encrypt the shuffled IDs, obtaining $\alpha H(\IDaa)$, where $H(\cdot)$ maps each ID to an elliptic-curve point.
(b) For $P_b$, it first shuffles its IDs with $\pi^b_2$ to obtain the shuffled IDs $\IDbb$. Next, $P_b$ generates an ECC key $\beta$ and uses $\beta$ to encrypt the shuffled IDs, obtaining $\beta H(\IDbb)$.

\noindent \textbf{Step 1: ECC-based ID Encrypting.}
\rone{To enable $P_a$ to later obtain the mapped intersection index pairs, the parties engage in an exchange of encrypted IDs using ECC, as follows:}
(1) $P_a$ sends the encrypted and shuffled IDs $\alpha H(\IDaa)$ to $P_b$.
(2) $P_b$ first shuffles the received data $\alpha H(\IDaa)$ with $\pi^b_1$ to obtain dual-shuffled IDs $\alpha H(\IDaab)$, where the composition permutation $\pi_1 = \pi^a_1 \circ \pi^b_1$. Next, using its private key $\beta$, $P_b$ encrypts the dual-shuffled IDs to obtain the dual-encrypted and dual-shuffled IDs $\beta\alpha H(\IDaab) = [\beta(x)$$]_{x \in \alpha H(\IDaab)}$. Finally, $P_b$ sends $\beta\alpha H(\IDaab)$ and its self-shuffled and encrypted IDs $\beta H(\IDbb)$ to $P_a$.
(3) $P_a$ decrypts the received data $\beta\alpha H(\IDaab)$ using its private key $\alpha$, obtaining $\beta H(\IDaab) = [\alpha^{-1}(x)]_{x \in \beta\alpha H(\IDaab)}$.
\rone{
Based on the hardness of the ECDLP, it is computationally infeasible for $P_a$ to recover $H(\IDaab)$ from $\beta H(\IDaab)$ or recover $H(\IDbb)$ from $\beta H(\IDbb)$.}
At this point, $P_a$ holds $\beta H(\IDaab)$ and $\beta H(\IDbb)$, both encrypted under $P_b$'s private key $\beta$.

\begin{figure}[htbp]
    \centering
    \subfigtopskip=0pt
    \subfigbottomskip=1pt
    \subfigcapskip=-4pt
    \includegraphics[width=0.36\textwidth]{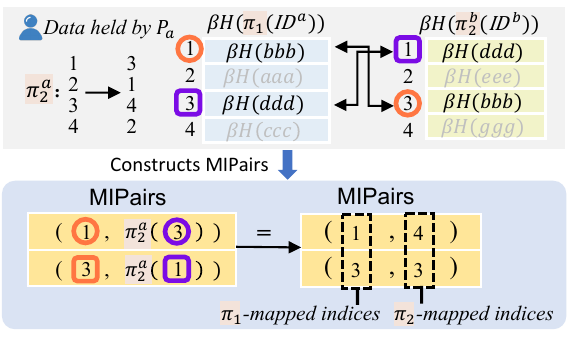}
    \caption{
    \rtwo{
    An illustration of constructing $\MII$ by $P_a$ in Step 2-(1) of $\Pi_{\text{SMIG}}$ (Protocol~\ref{pro:smig}), using the example data described at the beginning of Section~\ref{sec:design}.
    }
    }
    \label{fig:mig}
\end{figure}

\begin{small}
\begin{protocol}{$\Pi_{\rm \SMIG}$}{smig}
{
\textbf{Parameters:}
The input data row count $n$.
The ECC key bit length $\sigma$.
A hash function $H(\cdot)$ used for hashing an element to a point on the elliptic curve.
\\
\textbf{Input:} $P_a$ inputs identifiers $\IDa$. $P_b$ inputs identifiers $\IDb$.
\\
\textbf{Output:} Both $P_a$ and $P_b$ obtain the mapped intersection index pairs $\MII$.
\\
\textbf{Offline}:
$P_a$ samples two random permutations $\pi^a_1,\pi^a_2 \in [n]\mapsto[n]$. $P_b$ samples two random permutations $\pi^b_1,\pi^b_2 \in [n]\mapsto[n]$.
\\
\textbf{Setup}:
\begin{itemize} [leftmargin=1em,itemsep=0pt, topsep=0pt, partopsep=0pt, parsep=0pt]
\item[-] \textbf{[Shuffling and Encrypting Inputs of $P_a$]}. $P_a$ shuffles $\IDa$ with $\pi^a_1$ to obtain $\IDaa$. $P_a$ generates an ECC private key $\alpha$ and computes  $\alpha H(\IDaa) = [\alpha H(id)]_{id \in \IDaa}$.

\item[-] \textbf{[Shuffling and Encrypting Inputs of $P_b$]}. $P_b$ shuffles $\IDb$ with $\pi^b_2$ to obtain $\IDbb$. $P_b$ generates an ECC private key $\beta$ and computes $\beta H(\IDbb) = [\beta H(id)]_{id \in \IDbb}$.
\end{itemize}

\textbf{Online:}
\begin{enumerate}[label=\arabic{enumi}.,leftmargin=1.5em, partopsep=0pt, parsep=0pt,itemsep=0pt]
    \item \textbf{[ECC-based ID Encrypting]}.
    \begin{enumerate}[label=(\arabic{enumii}),leftmargin=1em,itemsep=0pt]
    \item $P_a$ sends $\alpha H(\IDaa)$ to $P_b$.
    \item $P_b$ shuffles $\alpha H(\IDaa)$ with $\pi^b_1$ to obtain dual-shuffled IDs $\alpha H(\IDaab)$, where the composition permutation $\pi_1 = \pi^a_1 \circ \pi^b_1$. Then, using its private key $\beta$, $P_b$ encrypts the dual-shuffled IDs to obtain the dual-encrypted IDs $\beta\alpha H(\IDaab) = [\beta(x)]_{x \in \alpha H(\IDaab)}$. $P_b$ sends this data $\beta\alpha H(\IDaab)$ and its encrypted and shuffled IDs $\beta H(\IDbb)$ to $P_a$.
    \item  $P_a$ decrypts the received data $\beta\alpha H(\IDaab)$ using its private key $\alpha$, obtaining $\beta H(\IDaab) = [\alpha^{-1}(x)]_{x \in \beta\alpha H(\IDaab)}$.
    \end{enumerate}
    \item \textbf{[Mapped Intersection Generation]}.
    \begin{enumerate}[label=(\arabic{enumii}),leftmargin=1.5em, topsep=0pt, partopsep=0pt, parsep=0pt]
    \item For $i\in[n]$, if $P_a$ can find some $j\in[n]$ such that $\beta H(\IDaab)_i == \beta H(\IDbb)_j$ (implying both are encrypted from the same intersection identifier $id$):
    \begin{enumerate}[label=-, leftmargin=1em,itemsep=0pt, topsep=0pt, partopsep=0pt, parsep=0pt]
    \item $P_a$ maps $j$ to $\pi^a_2(j)$.
    Let $i'$ denote the index of $id$ in $\IDb$, so that $j = \pi^b_2(i')$ and $\pi^a_2(j) = \pi^a_2(\pi^b_2(i')) = \pi_2(i')$.
    \item $P_a$ adds the pair $\left(i, \pi^a_2(j)\right)$ to $\MII$,
    where the first and second entries are the $\pi_1$-mapped index of $id$ in $\IDa$ and the $\pi_2$-mapped index of $id$ in $\IDb$, respectively.
    \end{enumerate}
    \item $P_a$ outputs $\MII$ (also sends $\MII$ to $P_b$).
    \end{enumerate}
\end{enumerate}
}
\end{protocol}
\end{small}

\noindent \textbf{Step 2: Mapped Intersection Generation.}
In this step, $P_a$ first obtains the mapped intersection index pairs $\MII$ from encrypted data $\beta H(\IDaab)$ and $\beta H(\IDbb)$, both of which are encrypted under $P_b$'s private ECC key $\beta$.
We provide a detailed explanation of $P_a$’s procedure for constructing $\MII$ in the next paragraph.
Secondly, $P_a$ outputs $\MII$ and sends $\MII$ to $P_b$.

\noindent \textit{$P_a$’s procedure for constructing $\MII$.}
As shown in Figure~\ref{fig:mig}, the core idea of this procedure relies on the observation that both $\beta H(\IDaab)$ and $\beta H(\IDbb)$ contain the same encrypted value $\beta H(id)$ for every intersection identifier $id$. Moreover, $\beta H(\IDaab)$ is securely shuffled under the composition permutation $\pi_1$, i.e., $\pi^a_1 \circ \rone{\pi^b_1}$, while $\beta H(\IDbb)$ is shuffled under $P_b$'s permutation $\pi^b_2$.
To generate the mapped intersection index pairs $\MII$, $P_a$ proceeds as follows.
For each $i \in [n]$ in $\beta H(\IDaab)$, if $P_a$ can find some $j \in [n]$ such that $\beta H(\IDaab)_i == \beta H(\IDbb)_j$,
\begin{itemize}[leftmargin=*,itemsep=0pt, topsep=0pt, partopsep=0pt, parsep=0pt]
    \item Index $i$ is exactly the $\pi_1$-mapped index of $id$ in $\IDa$, since encrypted IDs $\beta H(\IDaab)$ is $\pi_1$-shuffled.
    \item $P_a$ maps $j$ to $\pi^a_2(j)$. Let $i'$ denote the index of $id$ in $\IDb$, so that $j = \pi^b_2(i')$ and $\pi^a_2(j) = \pi^a_2(\pi^b_2(i')) = \pi_2(i')$, which is the $\pi_2$-mapped index of $id$ in $\IDb$.
    \item $P_a$ then appends the pair $\left(i, \pi^a_2(j)\right)$ to $\MII$, where the first entry is the $\pi_1$-mapped index of $id$ in $\IDa$ and the second entry is the $\pi_2$-mapped index of $id$ in $\IDb$.
\end{itemize}
Note that $P_a$ learns no information about the actual values of the intersection IDs, since $P_a$ only obtains encrypted and shuffled IDs.

\noindent \textit{\rtwo{Example:}} \rtwo{
As shown in Figure~\ref{fig:mig}, after step~1, $P_a$ obtains $\beta H(\IDaab)$ and $\beta H(\IDbb)$. Then, in step 2-(1), $P_a$ constructs the mapped intersection index pairs $\MII =[(1,4),(3,3)]$, which are then sent to $P_b$ in step 2-(2).}

\subsection{Secure Feature Alignment}
\label{subsec:fa}
In the secure feature alignment step, based on the outputs of $\Pi  _{\rm \SMIG}$ (Protocol~\ref{pro:smig}), $P_a$ and $P_b$ execute our proposed mapped intersection-based secure feature alignment protocol $\Pi_{\rm MI\text{-}SFA}$ to securely obtain the joined table shares $\share{\mathcal{D}}$.

As shown in Protocol~\ref{pro:sfa}, $\Pi_{\rm MI\text{-}SFA}$ takes as input feature columns $F^a$ from $P_a$ and feature columns $F^b$ from $P_b$. In addition, it takes as input the mapped intersection index pairs $\MII$ from both parties. After execution, $\Pi_{\rm MI\text{-}SFA}$ outputs to both parties joined table shares $\share{D}$.
$\Pi_{\rm MI\text{-}SFA}$ consists of an offline phase and an online phase.
In the offline phase, $P_a$ and $P_b$ generates random matrices required for the online execution of $\Pi_{\rm O\text{-}Shuffle}$ (Protocol~\ref{pro:oshuffle}).
The online phase consists of three steps.
In steps 1 and 2,  $P_a$ and $P_b$ securely dual-shuffle features from $P_a$ and $P_b$, respectively.
For dual-shuffling, each party performs the first-layer shuffle by locally using its permutation and performs the second-layer shuffle by executing $\Pi_{\rm O\text{-}Shuffle}$ (Protocol~\ref{pro:oshuffle}).
In step 3, both parties locally extract the shares of the joined table according to the mapped intersection index pairs $\MII$. Below, we introduce these steps in detail.

\noindent \textbf{Offline Phase.}
Each party $P_*$ ($* \in \{a,b\}$) generates a random matrix $R^* \in \mathbb{Z}^{n \times m_*}_{2^{\ell}}$.
Using the permutations $\pi^b_2$ and $\pi^a_2$ generated in the offline phase of $\Pi_{\rm \SMIG}$ (Protocol~\ref{pro:smig}) by $P_b$ and $P_a$, respectively, both parties involke the oblivious shuffle \(\mathcal{F}_{\textit{O-Shuffle}}\) to obtain the shares of shuffled random matrix, specifically $\share{\pi^b_2(R^a)}$, $\share{\pi^a_2(R^b)}$.

\par \noindent \textbf{Step 1: Dual-Shuffling ($\pi_1=\pi^a_1\circ\pi^b_1$) Features from $P_a$.}
Both parties securely obtain shares of $P_a$'s dual-shuffled features simply as follows:
(1) \(P_a\) locally shuffles its features \(F^a\) with its permutation \(\pi^a_1\) to obtain \(\pi^a_1(F^a)\).
(2) $P_a$ and $P_b$ jointly execute the online phase of $\Pi_{\rm O\text{-}Shuffle}$ (Protocol~\ref{pro:oshuffle}), where $P_a$ inputs shuffled features $\pi^a_1(F^a)$, random matrix $R^a$,  shuffled random matrix share $\share{\pi^b_1(R^a)}_a$ and receives dual-shuffled feature share $\share{\pi^b_1(\pi^a_1(F^a))}_a$, i.e., $\share{\pi_1(F^a)}_a$, while $P_b$ inputs its permutation $\pi^b_1$ (generated in the offline phase of Protocol~\ref{pro:smig}),  shuffled random matrix share $\share{\pi^b_1(R^a)}_b$ and receives dual-shuffled feature share $\share{\pi_1(F^a)}_b$.
Note that the only communication in the online phase of $\Pi_{\rm O\text{-}Shuffle}$ (Protocol~\ref{pro:oshuffle}) is $P_a$ sending its features masked with random matrix $R^a$ to $P_b$.

\par \noindent \textbf{Step 2: Dual-Shuffling ($\pi_2=\pi^b_2\circ\pi^a_2$) Features from $P_b$.}
In a symmetrical step analogous to the first step, the parties now dual-shuffle the features $F^b$ from $P_b$, with the roles of $P_a$ and $P_b$ exchanged. This process yields dual-shuffled shares $\share{\pi_2(F^b)}_a$ and $\share{\pi_2(F^b)}_b$ to $P_a$ and $P_b$, respectively.

\begin{small}
\begin{protocol}{$\Pi_{\rm MI\text{-}SFA}$}{sfa}
{
\textbf{Parameters:}
The input data row count $n$.
$P_a$'s feature dimension $m_a$.
$P_b$'s feature dimension $m_b$.
\\
\textbf{Input:} $P_a$ inputs feature columns $F^a$ and the mapped intersection index pairs $\MII$. $P_b$ inputs feature columns $F^b$ and the mapped intersection index pairs $\MII$.
\\
\textbf{Output:} $P_a$ and $P_b$ obtain the joined table shares $\share{\mathcal{D}}_a$ and $\share{\mathcal{D}}_b$, respectively, where $\mathcal{D}$ consists of the matched rows.
\\
\textbf{Offline:}
\begin{enumerate}[label=\arabic{enumi}.,leftmargin=0.8em,itemsep=0pt, topsep=0pt, partopsep=0pt, parsep=0pt]
        \item[-] \textbf{[Shuffling Random Matrix of $P_a$]}. $P_a$ samples random matrix $R^{a} \in \mathbb{Z}_{2^\ell}^{n \times m_a} $. $P_a$ and $P_b$ invoke $\mathcal{F}_{\textit{O-Shuffle}}$, where $P_a$ inputs $R^{a}$ and receives $\share{\pi^b_1(R^a)}_a$, while $P_b$ inputs permutation $\pi^{b}_1$, generated in the offline phase of $\Pi_{\rm \SMIG}$ (Protocol~\ref{pro:smig}), and receives $\share{\pi^b_1(R^a)}_b$.
        \item[-] \textbf{[Shuffling Random Matrix of $P_b$]}. $P_b$ samples random matrix $R^{b} \in \mathbb{Z}_{2^\ell}^{n \times m_b}$. $P_a$ and $P_b$ invoke $\mathcal{F}_{\textit{O-Shuffle}}$, where $P_a$ inputs permutation $\pi^a_2$, generated in the offline phase of $\Pi_{\rm \SMIG}$ (Protocol~\ref{pro:smig}), and receives $\share{\pi^a_2(R^b)}_a$, while $P_b$ inputs $R^{b}$ and receives $\share{\pi^a_2(R^b)}_b$.
\end{enumerate}
\textbf{Online:}
\begin{enumerate}[label=\arabic{enumi}.,leftmargin=1.2em,itemsep=0pt, topsep=0pt, partopsep=0pt, parsep=0pt]
\item \textbf{Dual-Shuffling ($\pi_1=\pi^a_1 \circ \pi^b_1$) Features from $P_a$}:
    \begin{enumerate}[label=(\arabic{enumii}),leftmargin=0.6em]
    \item  \( P_a \) shuffles its features \( F^a \) with \( \pi^a_1\) to obtain \( \pi^a_1(F^a) \).
    \item $P_a$ and $P_b$ execute online phase of $\Pi_{\rm O\text{-}Shuffle}$ (Protocol~\ref{pro:oshuffle}), where $P_a$ inputs $\pi^a_1(F^a)$, random matrix $R^a$, shuffled random matrix share $\share{\pi^b_1(R^a)}_a$ and receives dual-shuffled feature share $\share{\pi_1(F^a)}_a$, while $P_b$ inputs permutation $\pi^b_1$, generated in the offline phase of $\Pi_{\rm \SMIG}$ (Protocol~\ref{pro:smig}), shuffled random matrix share $\share{\pi^b_1(R^a)}_b$ and receives dual-shuffled feature share $\share{\pi_1(F^a)}_b$.
    \end{enumerate}
\item \textbf{Dual-Shuffling ($\pi_2=\pi^b_2 \circ \pi^a_2$) Features from $P_b$}:
    \begin{enumerate}[label=(\arabic{enumii}),leftmargin=0.6em]
    \item  \( P_b \) shuffles its features \( F^b \) with \( \pi^b_2\) to obtain \( \pi^b_2(F^b) \).
    \item $P_a$ and $P_b$ execute online phase of $\Pi_{\rm O\text{-}Shuffle}$ (Protocol~\ref{pro:oshuffle}), where $P_b$ inputs $\pi^b_2(F^b)$, random matrix $R^b$, shuffled random matrix share $\share{\pi^a_2(R^b)}_b$, and receives dual-shuffled feature share $\share{\pi_2(F^b)}_b$, while $P_a$ inputs permutation $\pi^a_2$, generated in the offline phase of $\Pi_{\rm \SMIG}$ (Protocol~\ref{pro:smig}), shuffled random matrix share $\share{\pi^a_2(R^b)}_a$ and receives dual-shuffled feature share $\share{\pi_2(F^b)}_a$.
    \end{enumerate}
\item \textbf{Locally Extract Intersection Data by the Mapped Intersection Indices}: Each party $P_*$ ( $* \in \{a,b\}$) extracts intersection data share from dual-shuffled feature shares according to $\MII$, yielding $\share{\mathcal{D}}_* = \{\share{\pi_1(F^a)_i}_* \| \share{\pi_2(F^b)_j}_*\}_{(i,j) \in \MII}$.
\end{enumerate}
}
\end{protocol}
\end{small}

\par \noindent \textbf{Step 3: Locally Extract Intersection Data by the Mapped Intersection.}
This step is straightforward.
Each party holds the dual-shuffled feature shares $\share{\pi_1(F^a)}$, $\share{\pi_2(F^b)}$ and mapped intersection index pairs $\MII$, where each pair are the $\pi_1$-mapped index of intersection identifier $id$ in $\IDa$ and the $\pi_2$-mapped index of intersection identifier $id$ in $\IDb$.
Thus, each party locally extracts shares from $\share{\pi_1(F^a)}$ whose indices are the first entry in each pair of $\MII$ and extracts shares from $\share{\pi_2(F^b)}$ whose indices are the second entry in each pair of $\MII$, forming its joined table share. That is, each party collects $\share{\mathcal{D}} = \{\share{\pi_1(F^a)_i} , \share{\pi_2(F^b)_j}\}_{(i,j) \in \MII}$.

\noindent \textit{\rtwo{Example:}}
\rtwo{
Recall that the mapped intersection index pairs $\MII = [(1,4),(3,3)]$ are obtained in Protocol~\ref{pro:smig}.
After steps~1 and~2, $P_a$ obtains dual-shuffled feature shares $\share{\pi_1(F^a)}_a = \langle [61,12,49,37] \rangle_a = [43, 8, 23, 9]$ and $\share{\pi_2(F^b)}_a = \share{[36,28,13,51]}_a = [25,15,4,16]$,
while $P_b$ obtains dual-shuffled feature shares $\share{\pi_1(F^a)}_b = \langle[61,12,$ $ 49,37]\rangle_b = [18, 4, 26, 28]$ and $\share{\pi_2(F^b)}_b = \share{[36,28,13,51]}_b = [11, $ $13, 9, 35]$.
In step 3, using $\MII$, each party locally collects $\share{\mathcal{D}} = \{\share{\pi_1(F^a)_i} , \share{\pi_2(F^b)_j}\}_{(i,j) \in \MII} = \share{[(61,51),(49,13)]}$. That is, $P_a$ obtains the joined table share $\share{\mathcal{D}}_a = [(43,16),(23,4)]$, while $P_b$ obtains the joined table share $\share{\mathcal{D}}_b = [(18,35),(26,9)]$.
}

\section{Security Proof}
\label{sec:proof}

Below, we prove the security of two protocols in \fname, namely $\Pi_{\rm \SMIG}$ (Protocol~\ref{pro:smig}) and $\Pi_{\rm MI\text{-}SFA}$ (Protocol~\ref{pro:sfa}), \rall{against semi-honest adversaries in the two-party setting (the definition of the security model is shown in Section~\ref{sec:preliminaries}).}

\begin{theorem}
\label{theorem:smig}
The \SMIG protocol $\Pi_{\rm \SMIG}$ (Protocol~\ref{pro:smig}) securely realizes $\mathcal{F}_{\textit{\SMIG}}$ in Figure~\ref{funct:smig} against semi-honest adversary $\mathcal{A}$.
\end{theorem}

\begin{proof}
We exhibit simulators $\mathsf{Sim_a}$ and $\mathsf{Sim_b}$ for simulating the view of the corrupt $P_a$ and $P_b$, respectively, and prove that the simulated view is indistinguishable from the real one via standard hybrid arguments.

\noindent \underline{Case 1 (Corrupt $P_a$).}
The simulator $\mathsf{Sim_a}$ receives $P_a$'s input $\IDa$ and $P_a$'s output $\MII$. $\mathsf{Sim_a}$ simulate the view of $\mathcal{A}$. The incoming messages to $\mathcal{A}$ is encrypted and shuffled IDs $\beta\alpha H(\IDaab)$ and $\beta H(\IDbb)$ received in Online step 1-(2).
$\mathsf{Sim_a}$ generates the view for $\mathcal{A}$ as follows:
\begin{itemize}[leftmargin=*,itemsep=0pt]
    \item $\mathsf{Sim_a}$ samples an ECC key $\beta$ from $\mathbb{Z}_{2^\sigma}$, two random permutations $\pi^b_1,\pi^b_2:[n] \mapsto [n]$. In addition, $\mathsf{Sim_a}$ constructs an $n$-sized IDs $\IDb$ by randomly sampling $c$ IDs from $P_a$'s IDs $\IDa$ and the remaining $n-c$ IDs from the complement of $\IDa$, where $c$ is the row count of the joined table, i.e., $|\MII|$. Moreover, $\mathsf{Sim_a}$ follows the real protocol to sample an ECC key $\alpha$ from $\mathbb{Z}_{2^\sigma}$, two random permutations $\pi^a_1,\pi^a_2:[n] \mapsto [n]$.
    \item $\mathsf{Sim_a}$ shuffles $\IDa$ with $\pi_1 = \pi^a_1 \circ \pi^b_1$ to obtain $\IDaab$. $\mathsf{Sim_a}$ computes $\beta\alpha H(\IDaab) = [\beta\alpha H(id)]_{id \in \IDaab}$.
    \item $\mathsf{Sim_a}$ shuffles $\IDb$ with $\pi^b_2$ to obtain $\IDbb$. $\mathsf{Sim_a}$ computes $\beta H(\IDbb) = [\beta H(id)]_{id \in \IDbb}$.
    \item $\mathsf{Sim_a}$ appends $\beta\alpha H(\IDaab)$ and $\beta H(\IDbb)$ to the view of $\mathcal{A}$.
    \item $\mathsf{Sim_a}$ follows the real protocol to construct output $\MII$ from $\alpha H(\IDaab)$ and $\alpha H(\IDbb)$.
    \item $\mathsf{Sim_a}$ appends $\MII$ to the view of $\mathcal{A}$.
\end{itemize}
We argue that the view output by $\mathsf{Sim_a}$ is indistinguishable from the real one. We first define three hybrid transcripts $T_0$, $T_1$, $T_2$, where $T_0$ is the real view of $P_a$, and $T_2$ is the output of $\mathsf{Sim_a}$.
\begin{enumerate}[label=\arabic{enumi}.,leftmargin=*,itemsep=0pt]
    \item $\textit{Hybrid}_0$. The first hybrid is the real interaction described in Protocol~\ref{pro:smig}. Here, an honest party $P_b$ uses real inputs and interacts with the corrupt party $P_a$. Let $T_0$ denote the real view of $P_a$.
    \item $\textit{Hybrid}_1$. Let $T_1$ be the same as $T_0$, except that $\beta\alpha H(\IDaab)$ and $\beta H(\IDbb)$ are replaced by the simulated values.
    By the security of ECC, the simulated $\beta\alpha H(\IDaab)$ and $\beta H(\IDbb)$ have the same distribution as they would in the real protocol. Hence, $T_0$ and $T_1$ are statistically indistinguishable.
    \item $\textit{Hybrid}_2$. Let $T_2$ be the same as $T_1$, except that $\MII$ is replaced by the simulated values.
    For each pair in $\MII$, the first and second entries are exactly the $\pi_1$-mapped index of the intersection identifier $id$ in $\IDa$ and the $\pi_2$-mapped index of the same intersection identifier $id$ in $\IDb$,
    In the real view, $\pi_1 = \pi^a_1 \circ \pi^b_1$ and $\pi_2 = \pi^b_2 \circ \pi^a_2$ are two random composition permutations,  which are unknown to $P_a$, their specific mappings are essentially random. Thus $\MII$ are uniformly random over $[n]$.
    The simulator generates random permutations $\pi^b_1$ and $\pi^b_2$.
    Even though $\pi^a_1$ and $\pi^a_2$ are fixed, applying a uniformly random permutation $\pi^b_1$ or $\pi^b_2$ afterward means the final mapping $\pi_1$ and $\pi_2$ are uniformly random.
    Thus, the resulting distributions of $\MII$ in both the real and simulated views are statistically identical. Therefore, the simulated view is indistinguishable from the real view.
    This hybrid is exactly the view output by the simulator.
\end{enumerate}

\noindent \underline{Case 2 (Corrupt $P_b$).}
The simulator $\mathsf{Sim_b}$ receives $P_b$'s input $\IDb$ and $P_b$'s output $\MII$. The incoming view of $\mathcal{A}$ consists of the encrypted and shuffled IDs $\alpha H(\IDaa)$ received in Online step 1-(1) and the output $\MII$ received in Online step 2-(2). $\mathsf{Sim_b}$ generates the view for $\mathcal{A}$ as follows:
\begin{itemize}[leftmargin=*,itemsep=0pt]
    \item $\mathsf{Sim_b}$ samples an ECC key $\alpha$ from $\mathbb{Z}_{2^\sigma}$, two random permutations $\pi^a_1, \pi^a_2:[n] \mapsto [n]$. In addition, $\mathsf{Sim_b}$ constructs an $n$-sized IDs $\IDa$ by randomly sampling $c$ IDs from $P_b$'s IDs $\IDb$ and the remaining $n-c$ IDs from the complement of $\IDb$, where $c$ is the row count of the joined table, i.e., $|\MII|$. Moreover, $\mathsf{Sim_b}$ follows the real protocol to sample an ECC key $\beta$ from $\mathbb{Z}_{2^\sigma}$, two random permutations $\pi^b_1,\pi^b_2:[n] \mapsto [n]$.
    \item $\mathsf{Sim_b}$ shuffles $\IDa$ with $\pi^a_1$ to obtain $\IDaa$. $\mathsf{Sim_b}$ computes $\alpha H(\IDaa) = [\alpha H(id)]_{id \in \IDaa}$.
    \item $\mathsf{Sim_b}$ appends $\alpha H(\IDaa)$ to the view of $\mathcal{A}$.
    \item $\mathsf{Sim_b}$ computes $\beta H(\IDaab)$ and $\beta H(\IDbb)$. Then, $\mathsf{Sim_b}$ involve the procedure for constructing $\MII$.
    \item $\mathsf{Sim_b}$ appends $\MII$ to the view of $\mathcal{A}$.
\end{itemize}

We argue that the view output by $\mathsf{Sim_b}$ is indistinguishable from the real one. We first define three hybrid transcripts $T_0$,$T_1$, $T_2$, where $T_0$ is the real view of $P_b$, and $T_2$ is the output of $\mathsf{Sim_b}$.
\begin{enumerate}[label=\arabic{enumi}.,leftmargin=*,itemsep=0pt]
    \item $\textit{Hybrid}_0$. The first hybrid is the real interaction described in Protocol~\ref{pro:smig}. Here, an honest party $P_a$ uses real inputs and interacts with the corrupt party $P_b$. Let $T_0$ denote the real view of $P_b$.
    \item $\textit{Hybrid}_1$. Let $T_1$ be the same as $T_0$, except that $\alpha H(\IDaa)$ is replaced by the simulated values. By the security of ECC, the simulated $\alpha H(\IDaa)$ has the same distribution as it would in the real protocol.
    \item $\textit{Hybrid}_2$. Let $T_2$ be the same as $T_1$, except that $\MII$ is replaced by the simulated values. Similar to Case 1, in the real view, $\pi_1 = \pi^a_1 \circ \pi^b_1$ and $\pi_2 = \pi^b_2 \circ \pi^a_2$ are two random permutations,  which are unknown to $P_a$, their specific mappings are essentially random. Thus, $\MII$ are uniformly random over $[n]$.
    The simulator generates random permutations to replace them. Thus, the resulting distributions of $\MII$ in both the real and simulated views are statistically identical.
    This hybrid is exactly the view output by the simulator. \qedhere
\end{enumerate}
\end{proof}

\begin{theorem}
\label{theorem:sfa}
The MI-SFA protocol $\Pi_{\rm MI\text{-}SFA}$ (Protocol~\ref{pro:sfa}) securely realizes $\mathcal{F}_{\textit{MI-SFA}}$ against semi-honest adversary $\mathcal{A}$.
\end{theorem}

\begin{proof}[Proof Sketch]
We argue the security of $\Pi_ {\rm MI\text{-}SFA}$ by reducing it to the security of its sub-protocol, $\Pi_{\rm O\text{-}Shuffle}$ (Protocol~\ref{pro:oshuffle}).
By design, $\Pi_{\rm MI\text{-}SFA}$ involves interaction only during execution of $\Pi_{\rm O\text{-}Shuffle}$. All other steps are performed locally and hence reveal no additional information to the other party.
\rone{
Intuitively,  $\Pi_{\rm O\text{-}Shuffle}$ is secure because its only communication step involves sending masked input data \(\hat{X} = X - R\), where \(R\) is a random matrix with entries sampled uniformly from \(\mathbb{Z}_{2^\ell}\). This ensures that \(\hat{X}\) is uniformly distributed and independent of \(X\), making any two inputs statistically indistinguishable from their masked versions. Consequently, a simulator can output a uniformly random matrix with distribution identical to the real execution.
As the security of $\Pi_{\rm O\text{-}Shuffle}$ is formally proven in~\cite{liu2023iprivjoin}, the security of Protocol~\ref{pro:sfa} follows.
}
\end{proof}
\section{Evaluation}
\label{sec:exp}

Our evaluation aims to answer the following questions:
\begin{itemize}[leftmargin=*,itemsep=0pt]
    \item Is it beneficial to use the redundancy-free secure data join protocol, \fname, instead of \CPSI~\cite{rindal2022blazing}, in the two-step SDA pipeline (secure data join and secure analytics) (Section~\ref{exp:e2e_cmp})?
    \item How does \fname perform relative to the SOTA redundancy-free secure data join protocol, \baseline~\cite{liu2023iprivjoin} (Section~\ref{exp:sda_cmp})?
\end{itemize}

\subsection{Implementation and Experiment Settings}
\label{sec:exp_set}
\noindent \textbf{Implementation.} We implement \fname using C++ based on the ECC library libsodium\footnote{\url{https://github.com/jedisct1/libsodium}}.
We represent all input data in fixed-point format over the ring $\mathbb{Z}_{2^{64}}$, as in the SDA frameworks~\cite{demmler2015aby,mohassel2017secureml}.
For ECC, we use Curve25519~\cite{bernstein2006curve25519}, which provides 128-bit security with a 256-bit key length. The curve equation of Curve25519 is $y^2 = (x^3 + 486662x^2 +x) \mod (2^{255} - 19) $.

\noindent \textbf{Experiment Environment.}
We perform all experiments on a single Linux server equipped with $2 \times$ 10-core 2.4 GHz Intel Xeon Silver 4210R and 1 TB RAM. Each party is simulated by a separate process with one thread.
We apply the tc tool\footnote{\url{https://man7.org/linux/man-pages/man8/tc.8.html}} to simulate a WAN setting with a bandwidth of 100 Mbps and 40ms round-trip time, following prior works~\cite{liu2023iprivjoin,mp-spdz}.

\noindent\textbf{Baseline.} We consider the following two baselines:
\begin{itemize}[leftmargin=*,itemsep=0pt, topsep=0pt, partopsep=0pt, parsep=0pt]
    \item  \baseline\cite{liu2023iprivjoin}: To the best of our knowledge, \baseline is the only existing protocol that provides the same functionality, i.e., a redundancy-free secure two-party data join protocol, as \fname.
    Since \baseline is not open-sourced, we re-implemented it, replacing its OPPRF libraries with faster ones~\cite{rindal2022blazing, rindal2021vole}.
    \item \CPSI~\cite{rindal2022blazing}: \CPSI~\cite{rindal2022blazing} is the SOTA circuit-based PSI protocol, which can be used for secure data join. \CPSI introduces many redundant dummy rows in the joined table (padded to the maximum length), other than the actual matched rows.
    We reproduce \texttt{CPSI}'s results using its open-source code\footnote{https://github.com/Visa-Research/volepsi.git}.
\end{itemize}
For OPPRF, we set the computational security parameter $\kappa = 128$ and the statistical security parameter $\lambda = 40$, consistent with the two baselines.
For Cuckoo hashing, we refer to the setting of \CPSI~\cite{rindal2022blazing}, with $\omega=1.27$ and 3 hashing functions.

\begin{table}[htbp]
\centering
\belowrulesep=0pt
\aboverulesep=0pt
\caption{Detailed information of datasets.}
\scalebox{0.86}{
\begin{tabular}{c|c|c}
\toprule
Dataset & Row Count ($n$) & Feature Dimension ($m$) \\
\toprule
Education Career Success & 5000 & 19 \\
Breast Cancer Gene & 1904 & 693 \\
Skin Cancer & 10000 & 785 \\
a9a & 32561 & 123 \\
ARCENE & 900 & 10000 \\
MNIST & 60000 & 784 \\
Tiny ImageNet & 100000 & 4096 \\
CelebA & 200000 & 5000 \\
\bottomrule
\end{tabular}
}
\label{tab:datasets}
\end{table}

\begin{table*}[t]
\belowrulesep=0pt
\aboverulesep=0pt
\centering
\caption{End-to-end online comparison of \fname, \baseline~\cite{liu2023iprivjoin} and \CPSI~\cite{rindal2022blazing} on two types of SDA tasks: secure statistical analysis and secure training (the term "secure" is omitted in the table). The costs of the SDA tasks for \fname and \baseline are identical, as both produce the same inputs for SDA tasks. The efficiency of secure training is evaluated in one epoch.
}
\makebox[\textwidth][c]{
\scalebox{0.86}{
\setlength{\tabcolsep}{2.2pt}
\begin{tabular}{c||c|c|cBc|c|c||c|c|cBc|c|c}
\toprule
\multicolumn{1}{c||}{}       & \multicolumn{6}{c||}{Education Career Sucess}                                         & \multicolumn{6}{c}{Breast Cancer Gene}
\\ \cmidrule{2-13}
Steps
                           & \multicolumn{3}{cB}{Time (s)}          & \multicolumn{3}{c||}{Communication (MB)}  & \multicolumn{3}{cB}{Time (s)}           & \multicolumn{3}{c}{Communication (MB)} \\ \cmidrule{2-13}
                           & {\small \fname}         & {\small \baseline\cite{liu2023iprivjoin}}    & {\small \CPSI\cite{rindal2022blazing}}    & {\small \fname}        & {\small \baseline\cite{liu2023iprivjoin}}    & {\small \CPSI\cite{rindal2022blazing}}       & {\small \fname}        & {\small \baseline\cite{liu2023iprivjoin}}     & {\small \CPSI\cite{rindal2022blazing}}     & {\small \fname}       & {\small \baseline\cite{liu2023iprivjoin}}     & {\small \CPSI\cite{rindal2022blazing}}      \\ \toprule
Secure Data Join      & 1.17         & 2.94         & 3.61    & 1.21        & 7.96         & 7.36       & 1.46        & 19.37         & 11.83    & 10.25      & 67.44         & 61.53     \\ \midrule
Statistical Analyze & \multicolumn{2}{c|}{3.04}    & 8.50    & \multicolumn{2}{c|}{0.50}   & 1.86       & \multicolumn{2}{c|}{11.19}   & 21.50    & \multicolumn{2}{c|}{103.22} & 186.63    \\ \midrule
Total                      & 4.21         & 5.98         & 12.11   & 1.71        & 8.46         & 9.22       & 12.65       & 30.56         & 33.33    & 113.47     & 170.66        & 248.16    \\ \hhline{=============}
\multicolumn{1}{c||}{}    & \multicolumn{6}{c||}{Skin Cancer}                                                 & \multicolumn{6}{c}{a9a}                                            \\ \cmidrule{2-13}
Steps
                           & \multicolumn{3}{cB}{Time (s)}          & \multicolumn{3}{c||}{Communication (MB)}  & \multicolumn{3}{cB}{Time (s)}           & \multicolumn{3}{c}{Communication (MB)} \\ \cmidrule{2-13}
                           & {\small \fname}         & {\small \baseline\cite{liu2023iprivjoin}}    & {\small \CPSI\cite{rindal2022blazing}}    & {\small \fname}        & {\small \baseline\cite{liu2023iprivjoin}}    & {\small \CPSI\cite{rindal2022blazing}}       & {\small \fname}        & {\small \baseline\cite{liu2023iprivjoin}}     & {\small \CPSI\cite{rindal2022blazing}}     & {\small \fname}       & {\small \baseline\cite{liu2023iprivjoin}}     & {\small \CPSI\cite{rindal2022blazing}}      \\ \toprule
Secure Data Join      & 7.11         & 108.07       & 73.63   & 60.87       & 385.51       & 319.05     & 8.24        & 61.88         & 33.62    & 33.98      & 247.85        & 228.79    \\ \midrule
Training & \multicolumn{2}{c|}{2371.44} & 3735.75 & \multicolumn{2}{c|}{683.14} & 1073.44    & \multicolumn{2}{c|}{7181.97} & 11365.90 & \multicolumn{2}{c|}{290.28} & 459.31    \\ \midrule
Total                      & 2378.55      & 2479.51      & 3809.38 & 744.01      & 1068.65      & 1392.49    & 7190.21     & 7243.85       & 11399.52 & 324.26     & 538.13        & 688.10  \\
\bottomrule
\end{tabular}
}
}
\label{tab:e2eapplication}
\end{table*}
\noindent \textbf{Datasets.}
Table~\ref{tab:datasets} summarizes all real-world datasets used in our evaluation.
To answer the first question (Section~\ref{exp:e2e_cmp}), we use four real-world datasets: Education Career Success~\cite{Education_and_Career_Success}, Breast Cancer Gene~\cite{Breast_Cancer_Gene_Expression}, Skin Cancer~\cite{tschandl2018ham10000}, and a9a~\cite{a9a}, to evaluate on secure data join and SDA tasks.
To answer the second question (Section~\ref{exp:sda_cmp}), we first evaluate \fname on eight datasets: the above four plus ARCENE~\cite{arcene}, MNIST~\cite{mnist}, Tiny ImageNet~\cite{le2015tiny}, and the large-scale CelebA~\cite{liu2015faceattributes}, following prior SDA studies~\cite{mohassel2017secureml,tan2021cryptgpu}.
These datasets vary in row counts (from $900$ to $200000$) and feature dimensions (from $19$ to $10000$).
\rtwo{For each dataset, we perform vertical partitioning, where the two parties hold disjoint feature columns but share the same row count as the original dataset.
We set the intersection rate to $\rho = 80\%$, consistent with \baseline~\cite{liu2023iprivjoin}. Specifically, we randomly select 80\% overlapping samples and fill the remaining 20\% of non-overlapping rows with randomly generated samples.
As analyzed in Appendix~\ref{app:comm_cmp}, the intersection rate ($\rho$) has a negligible impact on the performance of the baselines and \fname, so setting $\rho = 80\%$ provides a fair comparison.}
In addition, we configure that $P_a$ holds the data with a greater number of feature columns, as the baselines are more efficient in this setup compared to the reverse configuration (see Section~\ref{sec:exp_mm} for details).
Furthermore, we generate synthetic datasets with sizes up to 100 GB to evaluate efficiency under varying parameters.

\subsection{Two-step SDA Pipeline Evaluation}
\label{exp:e2e_cmp}

We compare \CPSI~\cite{rindal2022blazing}, \baseline~\cite{liu2023iprivjoin}, and \fname in a two-step SDA pipeline: secure data join and SDA task execution. The SDA task leverages the outputs of secure data join as its input data.

\noindent \textbf{SDA tasks.}
We evaluate two types of SDA tasks: secure statistical analysis and secure training.
(1) For secure statistical analysis, we examine the following two scenarios:
\begin{itemize}[leftmargin=*,itemsep=0pt, topsep=0pt, partopsep=0pt, parsep=0pt]
    \item Chi-square test: an education department and a tax department securely compute a chi-square test value on education data and income data~\cite{drosatos2011privacy}. To simulate this scenario, the feature columns in the Education Career Success dataset \cite{Education_and_Career_Success} are vertically partitioned into education-related and income-related features, which are held separately by the two parties.
    \item Pearson correlation: a hospital and a genomic company securely compute the Pearson correlation between disease and gene data \cite{lindell2020secure}.
    To simulate this scenario, the feature columns in the Breast Cancer Gene dataset~\cite{Breast_Cancer_Gene_Expression} are vertically partitioned into disease-related features and gene-related features, which are held separately by the two parties.
\end{itemize}
We design secure Chi-square test and secure Pearson correlation protocols (provided in Appendix~\ref{app:sec_statisical}).
Both are implemented using SMPC primitives in \MPSPDZ~\cite{mp-spdz}.
(2) For secure training, we examine two scenarios: One party holds users' image data while the other holds their labels; two parties hold disjoint subsets of users' features. In both scenarios, the parties aim to securely train a model on the joined table shares.
We use the Skin Cancer~\cite{tschandl2018ham10000} and a9a~\cite{a9a} datasets to simulate the above two scenarios, respectively.
We employ \MPSPDZ~\cite{mp-spdz}, a widely used open-source SMPC framework, to securely train logistic regression models over 10 epochs.

In terms of efficiency, Table~\ref{tab:e2eapplication} shows that \fname significantly accelerates the secure data join process, achieving up to $10.36\times$ faster running time and up to $6.73\times$ reduction in communication size compared to \CPSI, and up to $15.20\times$ faster running time and up to $7.29\times$ reduction in communication size compared to \baseline. The inferior performance of \baseline relative to \CPSI stems from its reliance on multiple rounds of oblivious shuffling to remove redundant rows present in \CPSI.
An important observation is that, during the downstream SDA process, \CPSI incurs an increase of $1.58 \sim 2.80\times$ in running time and $1.57 \sim 3.72\times$ in communication size.
The inefficiency of the SDA process caused by \CPSI arises from the redundant padded outputs in the join results.
Specifically, \CPSI outputs a joined table of row count \(N = \omega \cdot n\) (with \(\omega = 1.27\) in our experiments), where \(\rho \cdot n\) rows are intersection data (intersection rate $\rho$). For \(\rho = 80\%\), the redundant fraction of \CPSI's outputs is \(\bigl((\omega - \rho) \cdot n\bigr) / N \approx 37.01\%\).
In terms of correctness, Table~\ref{tab:e2eacc} shows that \CPSI leads to catastrophic error rate blowup in downstream secure statistical analysis and degrades the accuracy of secure training, with the Chi-square test statistic deviating by up to $13858.55\%$.
\rtwo{
Due to the limited performance and accuracy of \CPSI, we exclude it from subsequent secure data join evaluations.}

\begin{table}[htbp]
\belowrulesep=0pt
\aboverulesep=0pt
\centering
\caption{
Correctness of SDA tasks: \fname vs. \CPSI~\cite{rindal2022blazing} compared to performing data analytics on plaintext.
(`Education' = Education Career Success, `Breast' = Breast Cancer Gene ).
}
\scalebox{0.88}{
\setlength{\tabcolsep}{7pt}
\begin{tabular}{c||c|c||c|c}
\toprule
\textbf{SDA tasks}
        & \multicolumn{2}{c||}{\textbf{Stats. Percent Error}} & \multicolumn{2}{c}{\textbf{Training Accuracy Bias}} \\
Dataset & Education                 & Breast                 & Skin Cancer             & a9a                \\ \midrule
{\small \fname} & 0.00\%                     & 0.08\%                  & $-0.27\%$                 & $-0.42\%$
\\
{\small \CPSI\cite{rindal2022blazing}}    & 13858.55\%                 & 77.20\%                 & $-0.36\%$                 & $-0.50\%$            \\
\bottomrule
\end{tabular}
}
\label{tab:e2eacc}
\end{table}

\begin{figure*}[htbp]
    \centering
    \subfigtopskip=0pt
    \subfigbottomskip=0pt
    \subfigcapskip=-6pt
    \subfigure[Time (s)]{
        \includegraphics[width=0.31\textwidth]{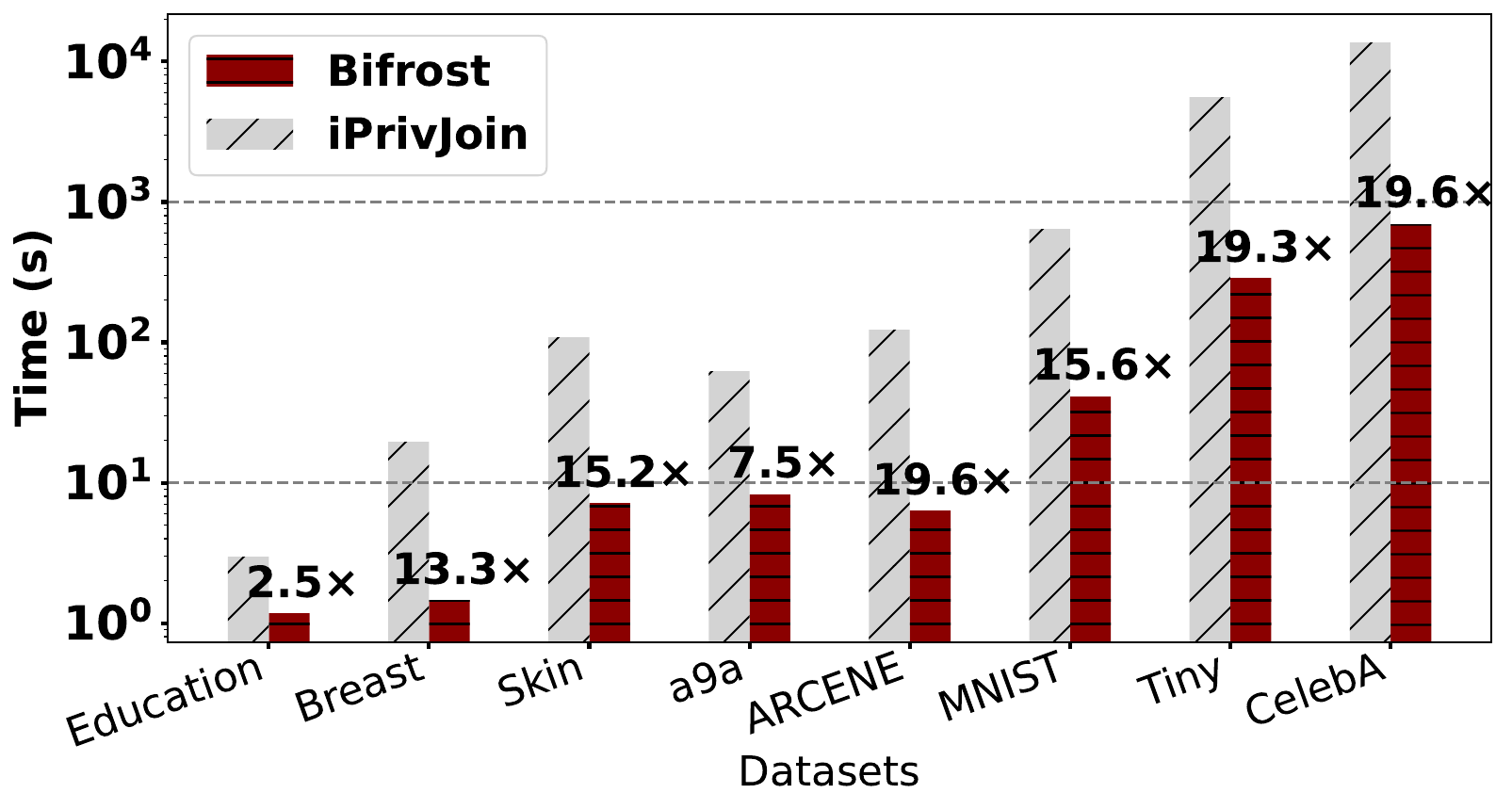}
    }
    \hfill
    \subfigure[Communication (MB)]{
        \includegraphics[width=0.31\textwidth]{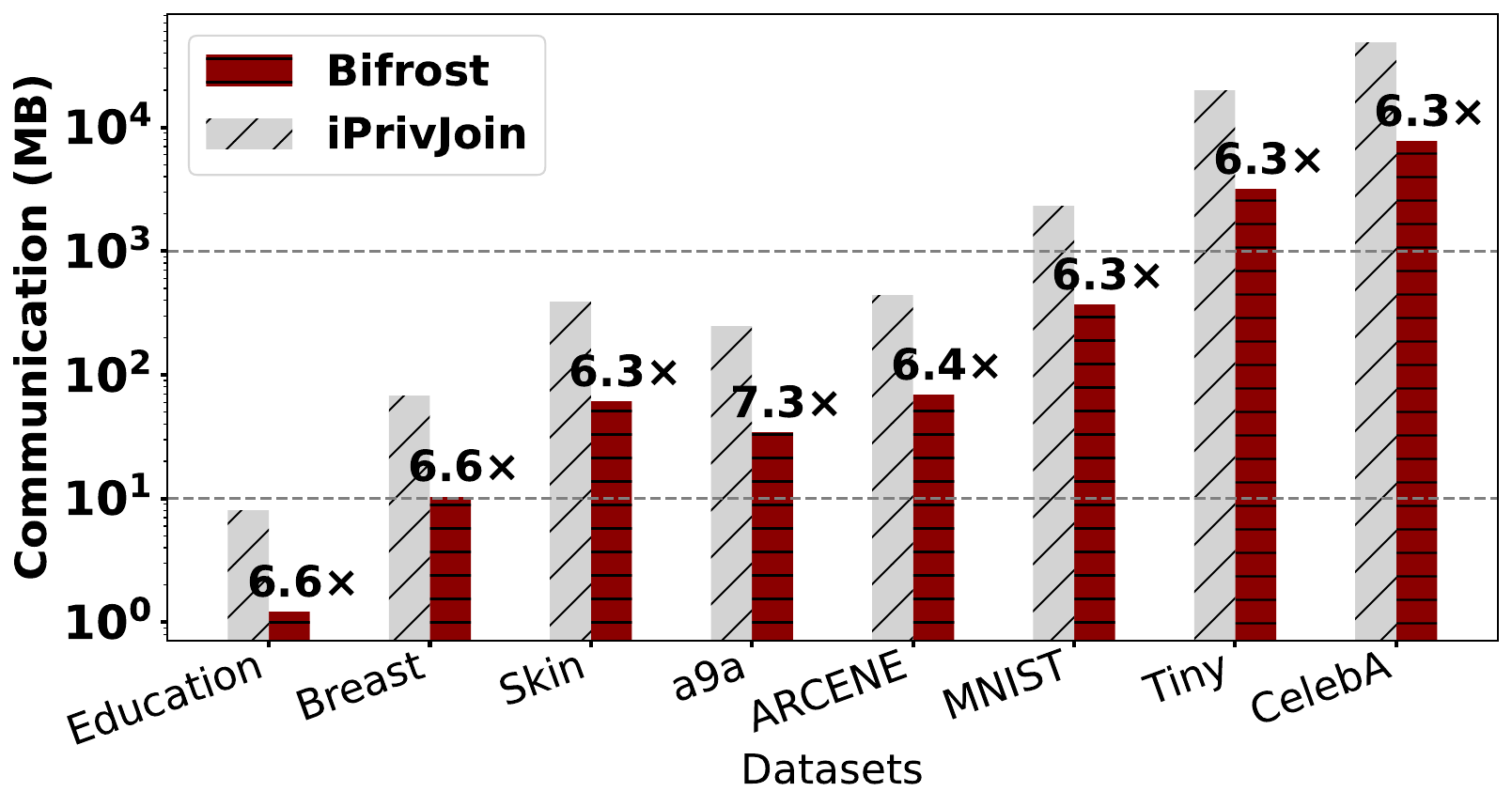}
    }
    \hfill
    \subfigure[Memory (MB)]{
        \includegraphics[width=0.31\textwidth]{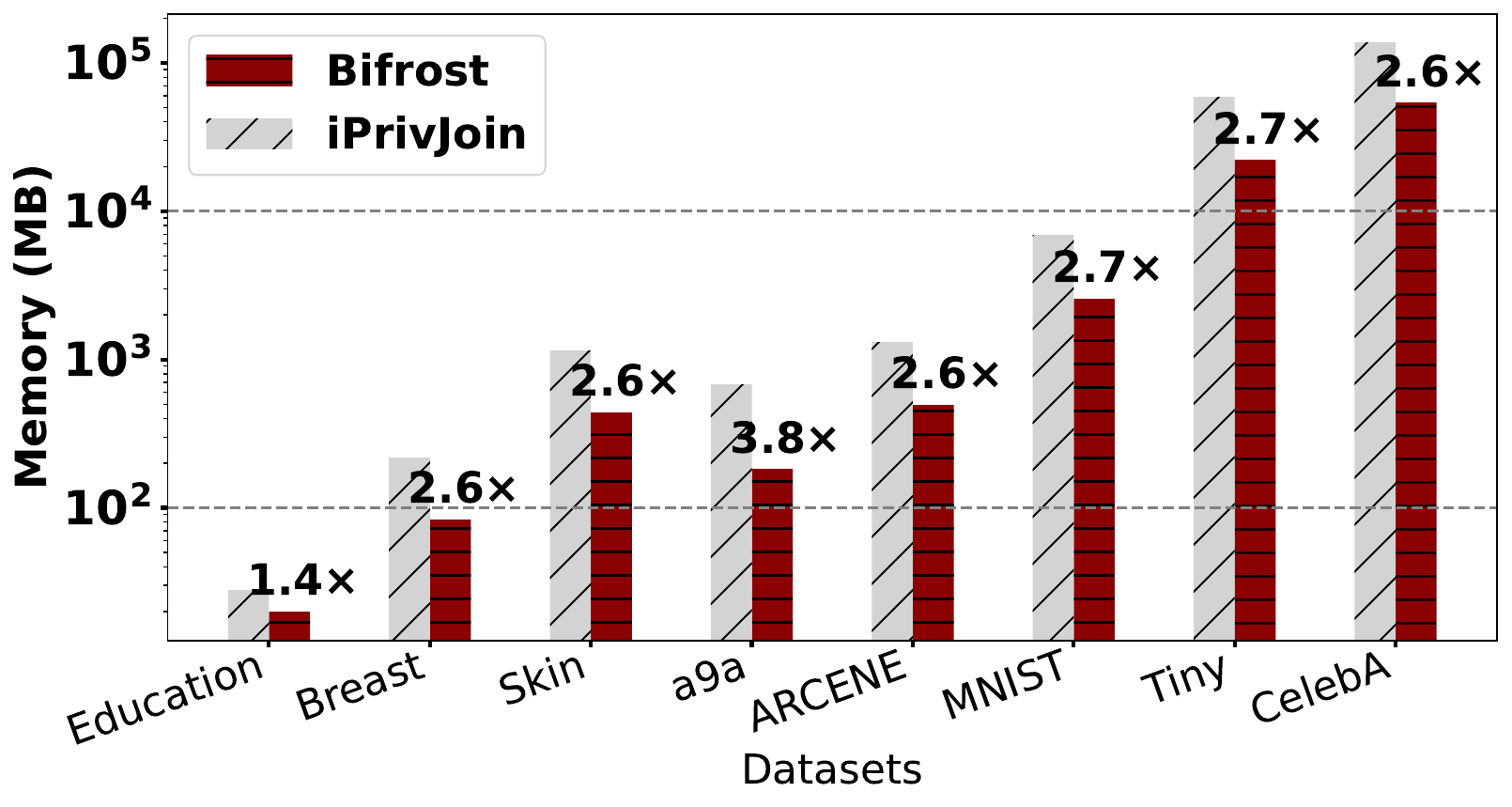}
    }
    \caption{Online Comparison of \fname vs. \baseline~\cite{liu2023iprivjoin} on eight real-world datasets.}
    \label{fig:e2ed}
\end{figure*}
\subsection{Secure Data Join Evaluation}
\label{exp:sda_cmp}

\subsubsection{Evaluation on Real-World Datasets across Three Phases.}
We compare \fname with \baseline~\cite{liu2023iprivjoin} across three phases (offline, setup, and online) using eight real-world datasets summarized in Table~\ref{tab:datasets}.
Detailed results are provided in Appendix~\ref{app:e_d}.
Table~\ref{tab:e2ed} and Figure~\ref{fig:e2ed} highlight the key results, summarized as follows:
\begin{itemize}
[leftmargin=*, itemsep=0pt, topsep=0pt, partopsep=0pt, parsep=0pt]
    \item In the online phase, \fname significantly outperforms \baseline by $2.54 \sim 19.64\times$ in running time.
    Moreover, \fname outperforms \baseline by $1.40 \sim 3.75\times$ in memory cost.
    The performance gain of \fname mainly stems from a reduction of $84.15\% \sim 86.29\%$ in communication size compared to \baseline.
    Specifically, the inefficiency of \baseline stems from using cryptographic primitive OPPRF and multiple rounds of oblivious shuffle for redundant data removal.
    In contrast, \fname requires a single round of oblivious shuffle while eliminating OPPRF. In addition, \fname requires only 4 communication rounds, fewer than the 11 rounds required by \baseline.
    \item Considering both the setup and online phases, \fname achieves an average speedup of $17.59\times$ over \baseline. Despite incurring high computational overhead in the setup phase due to ECC operations, \fname maintains superior overall performance. In addition, considering all three phases, \fname achieves an average speedup of $2.57\times$ in running time and a $2.59\times$ reduction in communication size. This improvement is primarily attributed to \fname's offline phase, which requires less than half of the communication size of \baseline.
    \item The performance advantage of \fname scales with dataset size: on the largest dataset, CelebA (15 GB), \baseline requires 227.20 minutes to complete the join, whereas \fname finishes in only 11.57 minutes, yielding a $19.64\times$ speedup in online running time, thereby highlighting its superior efficiency and scalability.
\end{itemize}

\begin{table}[htbp]
\belowrulesep=0pt
\aboverulesep=0pt
\caption{Average comparison of \fname vs. \baseline~\cite{liu2023iprivjoin} across three phases on eight real-world datasets. }
\scalebox{0.88}{
\setlength{\tabcolsep}{2pt}
\begin{tabular}{c||ccBcBcc}
\toprule
\multirow{2}{*}{Protocol} & \multicolumn{2}{cB}{Offline}                                                                                                                               & \multirow{2}{*}{\begin{tabular}[c||]{@{}c@{}}Setup\\ Time (s)\end{tabular}} & \multicolumn{2}{c}{Online}                                                                                                                              \\ \cline{2-6}
                          & Time (s)                                                                    & Comm. (MB)                                                                  &                                                                           & Time (s)                                                                   & Comm. (MB)                                                                 \\ \toprule
\multirow{2}{*}{{\small \fname  }}  & \textbf{11754.62} & \textbf{88307.52} & \multirow{2}{*}{12.28}                                                    & \textbf{130.86} & \textbf{1415.78} \\
                          &        (2.39$\times$)                                                                     &           (2.53$\times$)                                                                   &                                                                           &          (19.24$\times$)                                                                  &        (6.34$\times$)                                                                    \\ \hline
{\small \baseline\cite{liu2023iprivjoin}}                 & 28059.99                                                                    & 223704.32                                                                   & 0.22                                                                      & 2518.01                                                                    & 8983.04   \\
\bottomrule
\end{tabular}

}
\label{tab:e2ed}
\end{table}

In the following evaluation across varying parameters, we focus exclusively on the online-phase performance using datasets up to 100 GB.
Results on smaller-scale datasets are given in Appendix~\ref{app:e_small}, and offline-phase performance results are given in Appendix~\ref{app:e_off}.

\subsubsection{Evaluation across Different Row Counts.}
We conduct experiments on datasets of row count $n$ from $2^{14}$ to $2^{20}$ with a fixed feature dimension (\(m_a = m_b = 6400\)), which corresponds to dataset sizes scaling from 1.56 GB up to 100 GB.
As shown in Figure~\ref{fig:e2en}, \fname significantly outperforms \baseline by \(21.40 \sim 21.72\times\) in running time and \(7.70 \sim 7.73\times\) in communication size. Furthermore, the runtime advantage of \fname remains steady at \(21.59\times\), and its communication size reduction remains at about \(87.05\%\). This is because both \fname and \baseline have communication and computation costs that grow linearly in \(n\). Therefore, the relative efficiency advantage remains stable for fixed feature dimensions.

\begin{figure}[htbp]
    \centering
    \subfigtopskip=0pt
    \subfigcapskip=-6pt
    \subfigbottomskip=0pt
    \subfigure[Time (\rone{min})]{\includegraphics[width=0.23\textwidth]{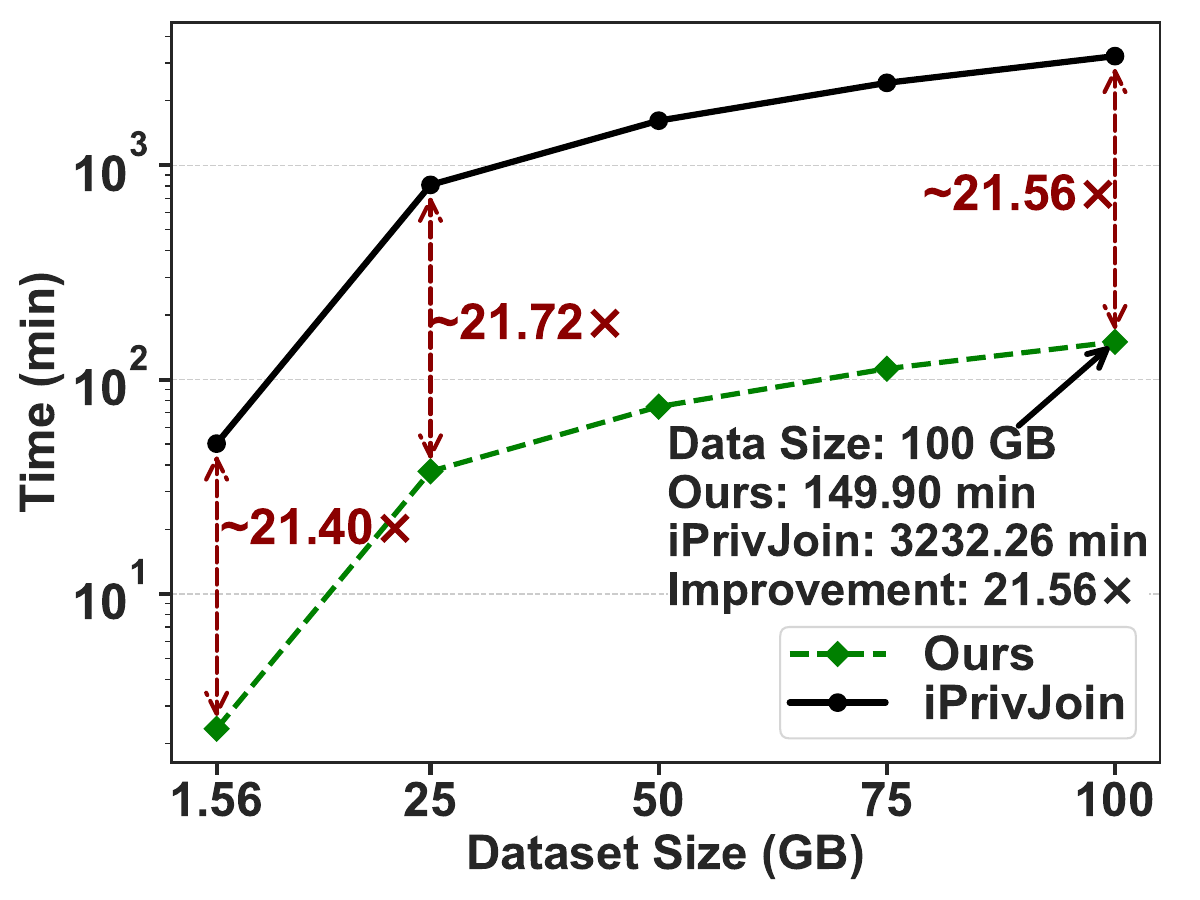}}
    \hfill
    \subfigure[Communication (\rone{GB})]{
    \includegraphics[width=0.23\textwidth]{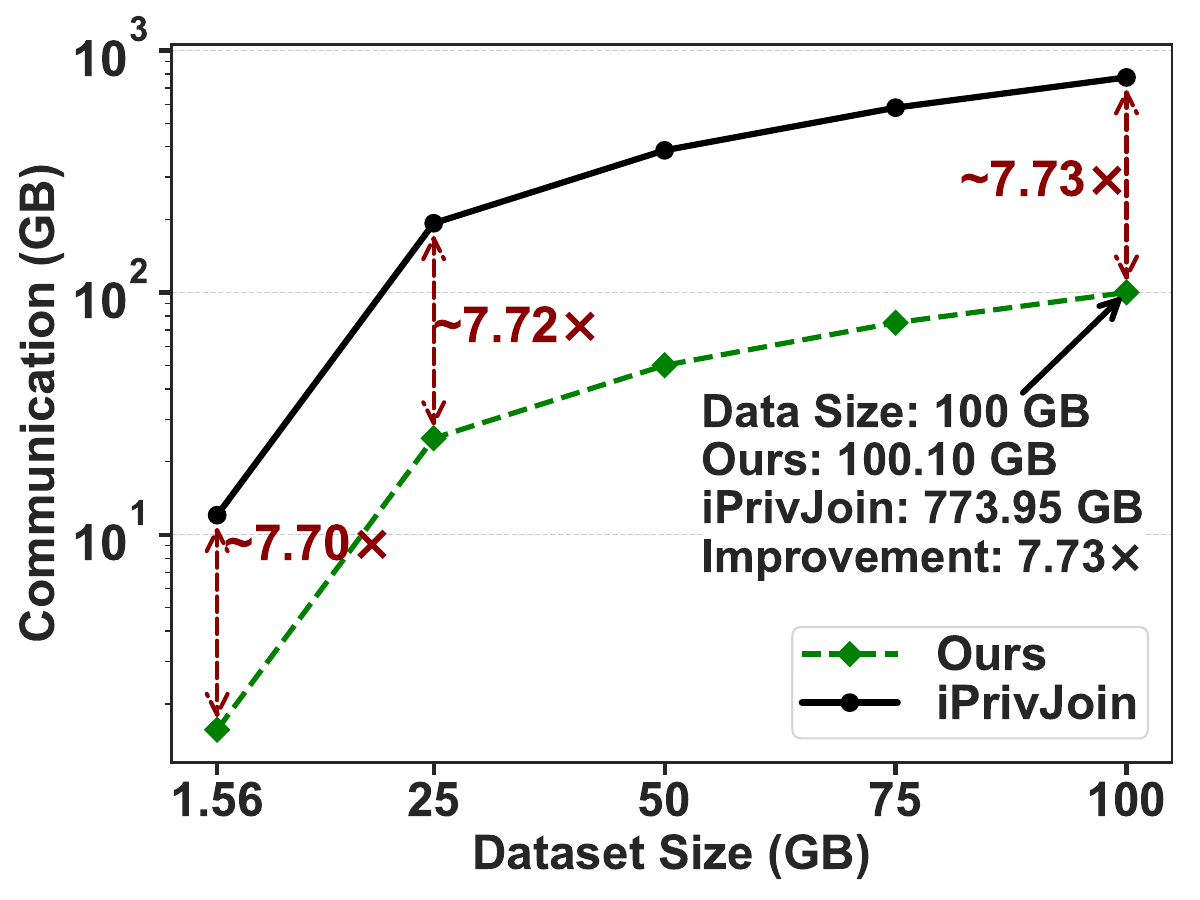}}
\caption{Comparison of \fname vs. \baseline~\cite{liu2023iprivjoin} on varying row counts when fixing feature dimension.
}
\label{fig:e2en}
\end{figure}

\subsubsection{Evaluation across Different Feature Dimensions.}
\label{sec:exp_mm}
We evaluate the effect of feature dimensionality under two settings: (1) We vary the feature dimensions $m_a = m_b$ from 100 to 6400 while fixing the dataset size at $n = 2^{20}$. (2) We vary one party’s feature dimension $m_*$ from 100 to 6400, where $* \in \{a,b\}$, while fixing the other party’s feature dimension to 100 and the dataset size to $n = 2^{20}$.
The corresponding results are presented in Figure~\ref{fig:e2em1} and Figure~\ref{fig:e2em2}, respectively.
We summarize as follows:
\begin{itemize}[leftmargin=*, itemsep=0pt, topsep=0pt, partopsep=0pt, parsep=0pt]
    \item \fname outperforms \baseline by $9.58 \sim 22.32\times$ in running time and $7.52\sim 9.07\times$ in communication size.
    In addition, the performance improvement of \fname becomes more pronounced as the feature dimension increases. For instance, as \(m_a = m_b\) increases from 100 to 6400, the running time improvement of \fname over \baseline increases from $9.58\times$ to $21.46\times$.
    This trend occurs because the computational and communication costs of \baseline grow more rapidly with the feature dimension than those of \fname. Therefore, \fname exhibits better efficiency when the feature dimension is large.
    \item
    \fname achieves better performance compared to \baseline, up to $22.32\times$ when increasing only $m_b$ and up to $19.52\times$ when increasing only $m_a$.
    The gap widens as $m_b$ grows because \baseline must share $P_b$’s features with $P_a$ via OPPRF, whereas \fname avoids this bottleneck. Thus \fname scales more efficiently when one party holds higher-dimensional features.
\end{itemize}

\begin{figure}[htbp]
    \centering
    \subfigtopskip=0pt
    \subfigbottomskip=0pt
    \subfigcapskip=-6pt
    \subfigure[Varying $m_a=m_b$ from 100 to 6400.]{\includegraphics[width=0.47\textwidth]{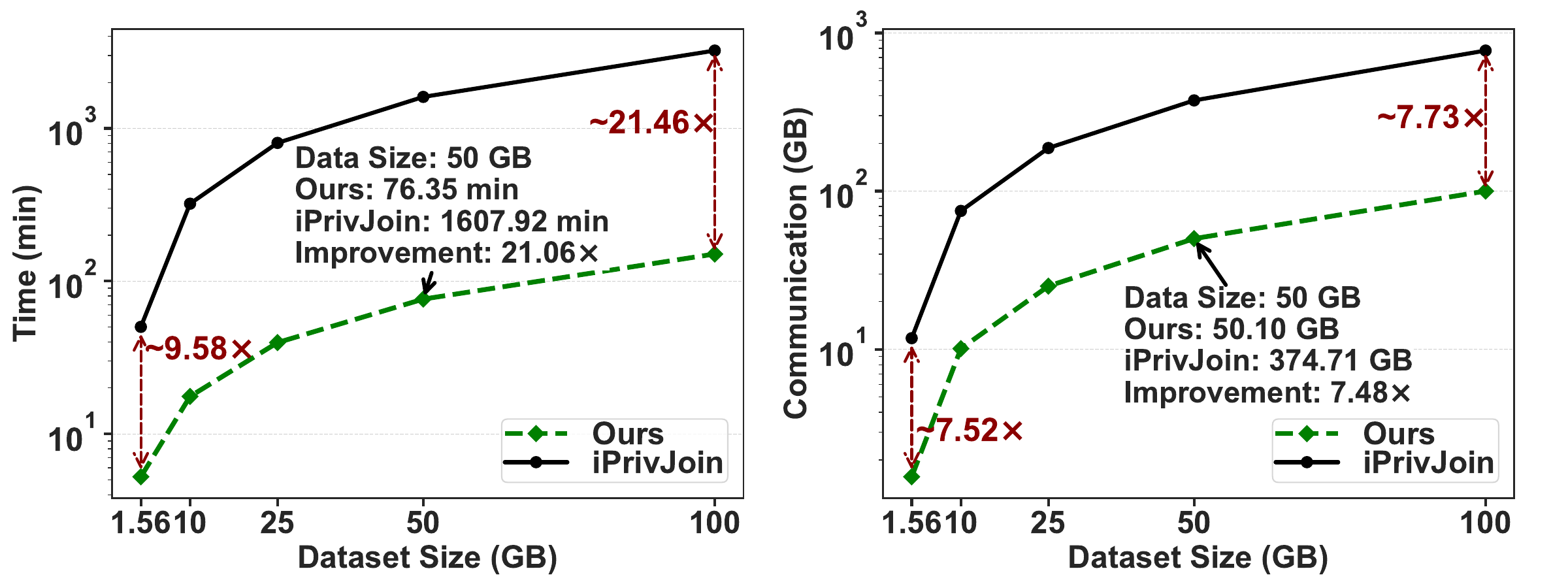}
    \label{fig:e2em1}
    }
    \subfigure[Varying $m_a$ (or $m_b$) from 100 to 6400 when fixing $m_b$ (or $m_a$).
    iPrivJoin\_${\rm m_a}$ and iPrivJoin\_${\rm m_b}$ denote the performance of \baseline when varying $m_a$ and $m_b$, respectively; Ours\_${\rm m_a/m_b}$ denotes the performance of \fname under either variation and is plotted as a single curve because $m_a$ and $m_b$ affect \fname identically.]{
    \includegraphics[width=0.48\textwidth]{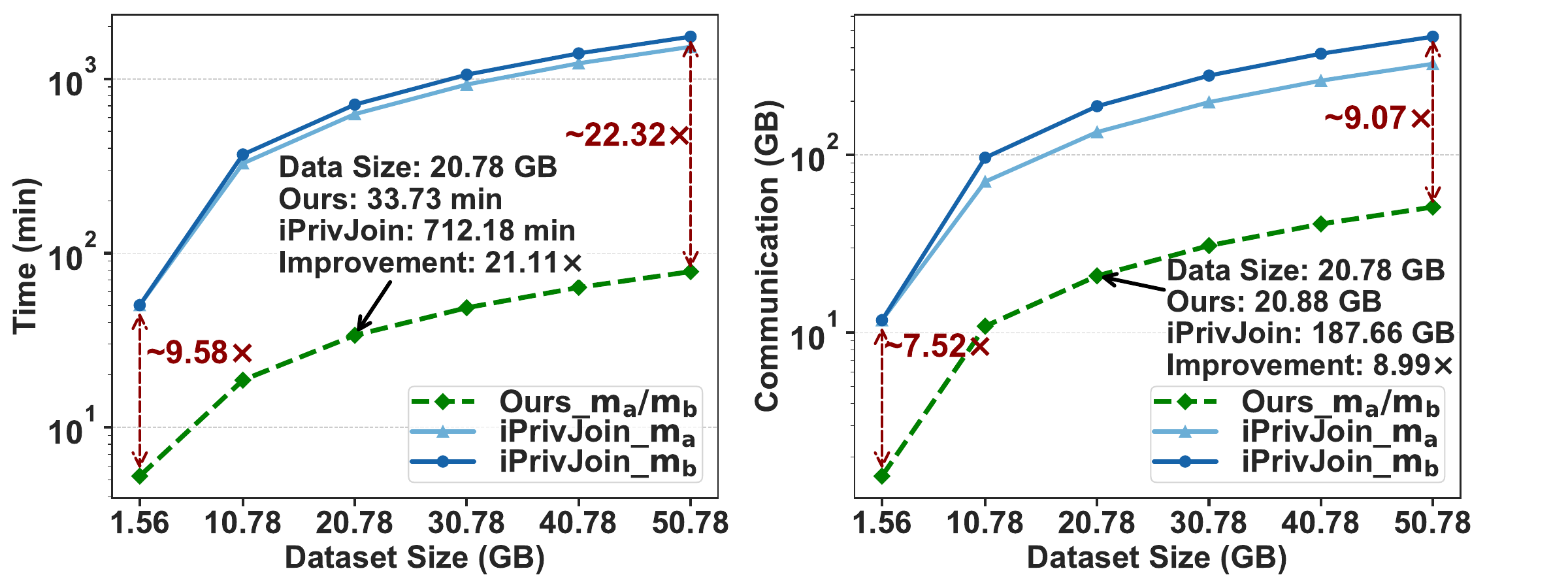}
    \label{fig:e2em2}
    }
\caption{Comparison of \fname vs. \baseline~\cite{liu2023iprivjoin} on varying feature dimensions when fixing row count.}
\label{fig:e2em}
\end{figure}

\subsubsection{Unbalanced Setting Experiments.}
One special scenario is the unbalanced setting~\cite{hao2024unbalanced,songsuda}, where the row count of the two parties' datasets differs by at least two orders of magnitude.
We conduct experiments with the SOTA secure unbalanced and redundancy-free data join protocol \suda~\cite{songsuda}.
The results demonstrate that \fname achieves $4.78 \sim 5.04\times$ faster running time compared to \suda in unbalanced settings
(see Appendix \ref{app:e_unbalance}).
\section{Discussion}
\label{sec:discussion}

\noindent \rall{\textbf{Supporting Join on Non-Key Attributes.}
For non-key attributes, there are two cases. In the first case, the attribute is unique within each table. The parties can treat this attribute as the identifiers and directly apply \fname to perform secure joins. In the second case, the attribute is non-unique, which means multiple records may share the same attribute value.
Handling non-unique attributes requires additional steps beyond the scope of this work, and we leave this as future work.
}

\noindent \rthree{
\textbf{Strengthening the Threat Model.}
\fname,  which is secure under the semi-honest model, can be strengthened to provide security under a stronger threat model that accounts for specific malicious protocol deviations.
We outline the strengthening path of \fname (comprising $\Pi_{\rm SMIG}$ and $\Pi_{\rm MI\text{-}SFA}$) toward this stronger threat model as follows: (i) For $\Pi_{\rm SMIG}$ (Protocol~\ref{pro:smig}), we can integrate a commitment scheme and Non-Interactive Zero-Knowledge proofs (NIZKs) to ensure that ECC-based ID encrypting is executed correctly, preventing malicious parties from using invalid keys; (ii) For $\Pi_{\rm MI\text{-}SFA}$ (Protocol~\ref{pro:sfa}), we can replace its underlying semi-honest oblivious shuffle protocol ($\Pi_{\rm O\text{-}Shuffle}$) with a malicious-secure shuffle protocol~\cite{song2023secret} to ensure that each party correctly shuffles the other party’s features. See Appendix~\ref{app:malicious} for a detailed discussion.
}

\section{Related work}
\label{sec:relatedwork}

\noindent\textbf{Circuit-based PSI.}
The initial \CPSI protocol is introduced by Huang \textit{et al.}~\cite{huang2012private}, which allows two parties, each with input $X$ and $Y$, to output the secret-shared $X\cap Y$ without revealing any other information.
Since the outputs are secret-shared, they can be directly used in subsequent MPC protocols, such as secure model training. In this way, \CPSI can be employed to perform secure data join.
Subsequent research on \CPSI primarily follows two directions: those based on OPPRF~\cite{kolesnikov2017practical,pinkas2019efficient,rindal2021vole,garimella2021oblivious,chandran2021circuit,rindal2022blazing} and others based on private set membership~\cite{ciampi2018combining,ma2022secure}.
A number of works~\cite{son2023psi,hao2024unbalanced,mahdavi2024pepsi} also focus on designing efficient \CPSI protocols for unbalanced settings.
Despite these innovations, \CPSI inherently introduces many redundant dummy rows (i.e., secret-shared zeros) in the joined table (padded to the maximum length), other than the actual matched rows. These redundant rows result in significant overhead to downstream secure data analytics tasks in many practical scenarios.

\par\vspace{0em}
\noindent \textbf{Balanced Secure Data Join Scheme.}
To address the above limitation in \CPSI,
Liu \textit{et al.} propose \baseline~\cite{liu2023iprivjoin}. To remove the redundant rows, \baseline introduces excessive communication overhead in the secure data join process, which stems from its reliance on OPPRF~\cite{rindal2021vole} and multiple rounds of oblivious shuffle.

\par\vspace{0em}
\noindent \textbf{Unbalanced Secure Data Join Scheme.}
Real-world scenarios may involve unbalanced datasets~\cite{songsuda}, where one party's dataset row count is significantly smaller than the other's.
Song \textit{et al.} propose the SOTA secure \textit{unbalanced} data join framework \suda~\cite{songsuda}. They leverage polynomial-based operations to efficiently output (the secret shares of) the redundancy-free joined table.
However, it scales inefficiently when directly applied to balanced settings, which are common and foundational settings.

\par\vspace{0em}
\noindent \textbf{Others.}
(1) Secure data join schemes~\cite{schnell2011novel,hardy2017private,24ElectronicsPPDA} for vertical federated learning rely on a Trusted Third Party (TTP), which may impose a strong trust assumption in practical scenarios.
(2) Private Set Union (PSU) scheme~\cite{sun2021vertical,yan2025idpriu} enables secure data join without a TTP but pads non-overlapping parts with synthetic data, resulting in redundant outputs.
(3) Enclave-based scheme~\cite{opauqe,efficienojoin,ObliDB} executes joins inside trusted hardware execution environments, which can reduce cryptographic overheads but inherit enclave trust.

\noindent We summarize representative secure data join schemes in Table~\ref{tab:techcmp}.

\begin{table}[htbp]
\caption{High-level comparison of various secure data join schemes with simplified communication complexity.
\rf{Here, $n$ is the row count of the input data; $m_a$ and $m_b$ are the feature dimensions for $P_a$ and $P_b$, respectively; $m = m_a + m_b$; $\ell$ and $\ell'$ are the bit-lengths of the element and the encrypted element, respectively; $\lambda$ and $\kappa$ are the statistical security parameter and the computational security parameter, respectively.}
}
\scalebox{0.82}{
\setlength{\tabcolsep}{1pt}
\begin{tabular}{ccc}
\toprule
Scheme       & Communication complexity    & Limitation                                   \\ \toprule
Ours         & $\mathcal{O}(nml)$                      & -                                            \\
\baseline~\cite{liu2023iprivjoin}    & $\mathcal{O}(nm_b\kappa + nm\ell)$ & High communication overhead                   \\
\CPSI~\cite{rindal2022blazing}          & $\mathcal{O}(nm_b\kappa + nm\ell)$ & Redundant rows                               \\
\suda~\cite{songsuda}        & $\mathcal{O}(nm\ell') $                     & Designed for unbalanced setting \\
FL-Join~\cite{hardy2017private}          & $\mathcal{O}(nm\ell')$                     & Relies on a trusted third party              \\
PSU~\cite{yan2025idpriu}          & $\mathcal{O}((n+n\log n)m\ell)$             & Redundant synthetic rows                     \\
Enclave-Join~\cite{efficienojoin} & $-$                           & Relies on enclave trust
\\
\bottomrule
\end{tabular}
}
\label{tab:techcmp}
\end{table}

\section{Conclusion}
In this paper, we propose an efficient and simple secure two-party data join protocol, referred to as \fname.
\fname outputs (the secret shares of) the redundancy-free joined table. The highlight of \fname lies in its simplicity: \fname builds upon two conceptually simple building blocks, an ECDH-PSI protocol and an oblivious shuffle protocol. The lightweight protocol design allows \fname to avoid many performance bottlenecks in the SOTA redundancy-free secure two-party data join protocol, \baseline, including the need for Cuckoo hashing and OPPRF. We also propose an optimization named \textit{dual mapping} that reduces the rounds of oblivious shuffle needed from two to one.
Compared to the SOTA redundancy-free protocol \baseline~\cite{liu2023iprivjoin}, \fname achieves $4.41 \sim 12.96\times$ faster running time and reduces communication size by $64.32\% \sim 80.07\%$.
Moreover, experiments on two-step SDA (secure join and secure analytics) show that \texttt{Bifrost}’s redundancy-free design avoids error blowup caused by redundant padded outputs (unlike \texttt{CPSI}), while delivering up to $2.80\times$ faster secure analytics with up to $73.15\%$ communication reduction compared to \CPSI.

\begin{acks}
 This work was supported by the National Cryptologic Science Fund of China (Number 2025NCSF01010) and the Natural Science Foundation of China (Number 92370120).

\end{acks}

\clearpage

\bibliographystyle{ACM-Reference-Format}
\bibliography{citations}

\section*{Appendix}

\renewcommand{\thesubsection}{\Alph{subsection}}

\subsection{Supporting Different Data Row Counts}
\label{app:difrowcounts}
In a real-world scenario, the two parties typically hold input data with different row counts.
We introduce a simple method to adapt our proposed \fname to support varying data count. Let $n_a$ and $n_b$ represent the input data row counts of parties $P_a$ and $P_b$, respectively.
The adaptation method is straightforward: both parties need to generate random permutations and a random matrix that corresponds to the sizes of their respective input data.
Specifically, in offline phase of $\Pi_{\rm \SMIG}$ (Protocol~\ref{pro:smig}), $P_a$ samples one random permutation $\pi^a_1 :[n_a] \mapsto [n_a]$ and another random permutation $\pi^a_2 :[n_b] \mapsto [n_b]$. Meanwhile, $P_a$ samples one random permutation $\pi^b_1 :[n_b] \mapsto [n_b]$ and another random permutation $\pi^b_2 :[n_a] \mapsto [n_a]$.
In addition, in the offline phase of $\Pi_{\rm MI\text{-}SFA}$ (Protocol~\ref{pro:sfa}), $P_a$ samples random matrix $R^a \in \mathbb{Z}^{n_a \times m_a}_{2^\ell}$, while $P_b$ samples random matrix $R^b \in \mathbb{Z}^{n_b \times m_b}_{2^\ell}$.

\subsection{Security Proof for Theorem~\ref{theorem:c}}
\label{app:proof}

\setcounter{theorem}{0}
\begin{theorem}
\label{theorem:c}
The functionality $\mathcal{F}_{\textit{\SMIG}}$ in Figure~\ref{funct:smig} reveals no information beyond what is revealed by the joined table row count $c$ to both parties.

\begin{proof}[Proof Sketch.]
Let $\mathcal{A}$ be a PPT adversary corrupting either $P_a$ or $P_b$, and let $c = |\mathsf{MIPairs}|$ denote the size of the intersection (i.e., the row count of the joined table). The simulator $\mathsf{Sim}$ is given only $c$. $\mathsf{Sim}$ samples two lists of independent, uniformly random values $I^1$ and $I^2$, each of size $c$, and constructs the set of pairs $\mathsf{MIPairs}^{\mathcal{S}} = [(I^1_i, I^2_i)]_{i \in [c]}$. $\mathsf{Sim}$ then outputs $\mathsf{MIPairs}^{\mathcal{S}}$ as the simulated view of adversary $\mathcal{A}$.
By the uniform randomness of the unknown permutations $\pi_1$ and $\pi_2$ from each party, the distribution of each mapped index in $\mathsf{MIPairs}$ is uniform and independent.
Therefore, the distribution of $\mathsf{MIPairs}^{\mathcal{S}}$ is identical to that of $\mathsf{MIPairs}$ observed by $\mathcal{A}$ in the execution of $\mathcal{F}_{\textit{SMIG}}$. Hence, for any adversary $\mathcal{A}$, the simulated view is indistinguishable from the real one.
\end{proof}

\end{theorem}

\subsection{Communication Complexity Comparison}
\label{app:comm_cmp}
\noindent \textbf{Online communication analyze of \fname.}
The online communication cost of \fname involves 4 rounds and $3n\sigma + 2c\lceil\log_2 n\rceil + nm\ell$ bits, where $\sigma$ is the bit-length of the ECC private key.
The online communication of \fname consists of two parts: SMIG protocol $\Pi_{\rm \SMIG}$ (Protocol~\ref{pro:smig}) and MI-SFA protocol $\Pi_{\rm MI\text{-}SFA}$ (Protocol~\ref{pro:sfa}).
\begin{itemize}[leftmargin=*,itemsep=0pt, topsep=0pt, partopsep=0pt, parsep=0pt]
    \item The SMIG protocol has a communication complexity of 3 rounds and a size of $3n\sigma + 2c\lceil\log_2 n\rceil$ bits. The process begins with $P_a$ sending $n$ encrypted elements to $P_b$, resulting in $n\sigma$ bits. In the second round, $P_b$ sends two lists of $n$ encrypted elements back to $P_a$, resulting in $2n\sigma$ bits. In the final round, $P_a$ sends the mapped intersection index pairs to $P_b$, resulting in $2c\lceil\log_2 n\rceil$ bits, where $c$ is the row count of the joined table.
    \item  The MI-SFA protocol has a communication complexity of 1 round and a size of $n(m_a+m_b)\ell$, assuming parallel execution of its send and receiving operations. The communication cost is entirely attributed to two concurrent invocations of the $\Pi_{\rm O\text{-}Shuffle}$ sub-protocol. In this round, party $P_a$ sends its $n \times m_a$ masked feature matrix, a total of $nm_a\ell$ bits, while $P_b$ sends its $n \times m_b$ masked feature matrix, a total of $nm_b\ell$ bits.
\end{itemize}
The SMIG protocol cost is independent of the feature dimensions $m$ and is much smaller than that of the MI-SFA protocol. Consequently, the overall communication of \fname is dominated by the MI-SFA protocol and scales roughly with the dataset size $n\cdot m$.

\par\vspace{0.5em}
\noindent \textbf{Offline communication analyze of \fname.}
\fname requires the shuffled random matrix shares generated by $\Pi_{\rm O-Shuffle}$ (Algorithm 2 in \baseline\cite{liu2023iprivjoin}) to mask the feature data of $P_a$ and $P_b$, respectively. Therefore, the communication size of our offline phase is \(\mathcal{O}((\kappa + \lambda)n\log nm)\) bits.

\par\vspace{0.5em}
\noindent \textbf{Online communication analyze of \baseline~\cite{liu2023iprivjoin}.} \baseline consists of three steps: private data encoding, oblivious shuffle, and private data trimming.
In the first step, both parties invoke the OPPRF on the IDs and features to share $P_b$'s features. $P_a$ secret shares its feature data.
And two parties invoke an OPRF once to encode ID data.
Thus, both parties obtain an encoded ID vector and a secret-shared dataset containing aligned features and redundant samples.
The communication size for the first step is $\mathcal{O}(hn(\lambda + \log n + m_b\kappa) + nm_a\ell)$ bits, and the number of rounds is 8. In the second step, both parties send and receive twice the masked datasets, which results in a communication size of $\mathcal{O}(nm\ell + n\kappa)$ and 2 communication rounds. In the final step, the encoded IDs are reconstructed, resulting in the communication size of $\mathcal{O}(n\kappa)$, and the number of rounds is 1. Thus, the total communication complexity and the number of rounds for \baseline are $\mathcal{O}(hn(\lambda + \log n + m_b\kappa) + nm\ell + n\kappa)$ bits and 11, respectively.

\par\vspace{0.5em}
\noindent \textbf{Offline communication analyze of \baseline~\cite{liu2023iprivjoin}.} \baseline requires the shuffled random matrix shares generated by two times of $\Pi_{\rm O-Shuffle}$ (Algorithm 2 in \baseline \cite{liu2023iprivjoin}) to mask both feature data and ID data from both parties. As a result, the communication size in the offline phase is \(\mathcal{O}((\omega\kappa + \lambda)n\log(\omega n) m)\)bits.

\subsection{Two Protocols for Secure Statistical Analysis}
\label{app:sec_statisical}
The full protocols we design for the secure Chi-square test (i.e., $\mathcal{X}^2$-test) computation and the secure Pearson correlation computation are shown in Protocol~\ref{pro:x2test} and Protocol~\ref{pro:pearson}, respectively. Below, we introduce the two protocols.

\subsubsection{Secure Chi-square Test}
\begin{small}
\begin{protocol}{$\Pi_{\rm \mathcal{X}^2\text{-}Test}$}{x2test}
\smallskip
\textbf{Parameters:}
The categories of the first feature $\mathit{CatX}$;
The categories of the second feature $\mathit{CatY}$;
The row count of input data $n$;
\smallskip
\\
\textbf{Inputs:} Two columns of the secret-shared aligned features $\share{X}$ and $\share{Y}$.
\\
\textbf{Outputs:} A secret-shared $\mathcal{X}^2\text{-}test$ value $\share{v}$, where $v = \mathcal{X}^2\text{-}test(X,Y)$

\begin{algorithmic}[1]

    \FOR{$j = 1$ \TO $|CatX|$ in parallel}
        \FOR{$k = 1$ \TO $|\mathit{CatY}|$ in parallel}
            \STATE $\share{\mathit{CatXY}_{i,j+k*|\mathit{CatX}|}} = (\share{X_i\|Y_i} == \mathit{CatX}_j\|\mathit{CatY}_k)$.
        \ENDFOR
    \ENDFOR

    \FOR{$i = 1$ \TO $n$ in parallel}
        \FOR{$j = 1$ \TO $|\mathit{CatX}|$ in parallel}
            \FOR{$k = 1$ \TO $|\mathit{CatY}|$ in parallel}
                \STATE $\share{\mathit{ContingencyT}_{j,k}} = \share{\mathit{ContingencyT}_{j,k}} + \share{\mathit{CatXY}_{i,j+k*|\mathit{CatX}|}}$.
            \ENDFOR
        \ENDFOR
    \ENDFOR

    \FOR{$i = 1$ \TO $|\mathit{CatX}|$ in parallel}
        \FOR{$j = 1$ \TO $|\mathit{CatY}|$ in parallel}
            \STATE $\share{\mathit{ColSum}_i} = \share{\mathit{ColSum}_i} + \share{\mathit{ContingencyT}_{i,j}}$.
            \STATE  $\share{\mathit{RowSum}_j} = \share{\mathit{RowSum}_j} + \share{\mathit{ContingencyT}_{i,j}}$.
        \ENDFOR
    \ENDFOR

    \STATE $\share{\mathit{RowSum}}= \share{\mathit{RowSum}}/n$.

    \STATE $\share{E} = \mathit{MatMul}(\share{\mathit{ColSum}_i}, \share{\mathit{RowSum}}^T)$.

    \STATE $\share{v} = \mathit{Sum}\left((\share{\mathit{ContingencyT}} - \share{E})^2/\share{E}\right)$
\end{algorithmic}

\end{protocol}
\end{small}
As shown in Protocol~\ref{pro:x2test}, the secure $\mathcal{X}^2$-test computation protocol $\Pi_{\rm \mathcal{X}^2\text{-}test}$ inputs two columns of secret-shared aligned features $\share{X}$ and $\share{Y}$, and outputs a secret-shared $\mathcal{X}^2$-test vale $\share{v}$, where $v = \mathcal{X}^2\text{-}test(X,Y)$.
In addition, $\mathit{CatX}$ are the categories of the feature $X$ and $\mathit{CatY}$ are the categories of the feature $Y$.

The protocol begins by having the parties jointly construct a secret-shared contingency table $\share{\mathit{ContingencyT}}$, where each entry $\share{\mathit{ContingencyT}_{j,k}}$ contains the count of records for which the features correspond to the category pair $(\mathit{CatX}_j\|\mathit{CatY}_k)$ (Lines 1-12). Next, the parties compute the secret-shared marginal totals for rows ($\share{\mathit{ColSum}}$) and columns ($\share{\mathit{RowSum}}$) from the contingency table (Lines 13-18). Using these marginals, they then calculate the secret-shared expected frequency table $\share{E}$ under the null hypothesis of independence (Lines 19-20). Using these secret-shared marginals, they then calculate the secret-shared expected frequency table $\share{E}$ under the null hypothesis of independence (Lines 19-20). To mitigate numerical growth, our protocol introduces an optimization: instead of computing the matrix product of the marginals and then dividing by the total count n, we first divide the $\share{\mathit{RowSum}}$ vector by $n$ (Line 19) before the matrix multiplication (Line 20). This ordering is crucial for computations over finite rings, as it reduces the magnitude of the intermediate values, thereby preventing potential overflows and enabling the use of a smaller ring, which in turn minimizes computational and communication overhead.
Finally, the protocol computes the $\mathcal{X}^2$-test share $\share{v}$ by securely summing the element-wise squared differences between the observed frequencies in $\share{\mathit{ContingencyT}}$ and the expected frequencies in $\share{E}$, normalized by the expected frequencies (Line 21).

\subsubsection{Secure Pearson Correlation}
\begin{small}
\begin{protocol}{$\Pi_{\rm Pearson}$}{pearson}
\smallskip
\textbf{Parameters:}
The row count of input data $n$;
The bit length of element $\ell$.
\smallskip
\\
\textbf{Inputs:} Two columns of the secret-shared aligned features $\share{X}$ and $\share{Y}$.
\\
\textbf{Outputs:} A secret-shared Pearson correlation value $\share{v}$, where $v = Pearson(X,Y)$.
\smallskip
\\
\textbf{Online:}
\begin{algorithmic}[1]
\STATE $\share{\mathit{X\_mean}} = Sum(\share{X})/n$.
\STATE $\share{\mathit{Y\_mean}} = Sum(\share{Y})/n$.

\FOR{$i = 1$ \TO $n$ in parallel}
    \STATE $\share{\mathit{X\_dif}_i} = \share{X_i} - \share{\mathit{X\_mean}}$.
    \STATE $\share{\mathit{Y\_dif}_i} = \share{Y_i} - \share{\mathit{Y\_mean}}$.
\ENDFOR

\STATE $\share{covariance} = \mathit{DotProduct}(\share{\mathit{Y\_dif}}, \share{\mathit{X\_dif}})$.
\STATE $\share{X\_dev} = \mathit{DotProduct}(\share{\mathit{X\_dif}}, \share{\mathit{X\_dif}})$.
\STATE $\share{Y\_dev} = \mathit{DotProduct}(\share{\mathit{Y\_dif}}, \share{\mathit{Y\_dif}})$.

\STATE $\share{X\_dev\_inv} = \mathit{InvertSqrt}(\share{X\_dev})$.
\STATE $\share{Y\_dev\_inv} = \mathit{InvertSqrt}(\share{Y\_dev})$.

\STATE $\share{v}= \share{covariance} * \share{X\_dev\_inv} * \share{Y\_dev\_inv}$.
\end{algorithmic}
\end{protocol}
\end{small}
As shown in Protocol~\ref{pro:pearson}, the secure Pearson correlation computation protocol $\Pi_{\rm Pearson}$ inputs two columns of secret-shared aligned features $\share{X}$ and $\share{Y}$, and outputs a secret-shared pearson correlation value $\share{v}$, where $v = Pearson(X,Y)$.

\begin{table*}[htbp]
\caption{Comparison of \fname vs. \baseline~\cite{liu2023iprivjoin} across three phases (offline, setup, and online) on eight real-world datasets. Results from the Online phase are highlighted: blue indicates better performance, while gray indicates worse performance. ``Education'' refers to the Education Career Success dataset.
}
\belowrulesep=0pt
\aboverulesep=0pt
\scalebox{0.88}{
\begin{tabular}{ccccccccccc} \toprule
\multicolumn{3}{c}{Dataset} & Education & Breast Cancer Gene & Skin Cancer & a9a & ARCENE & MNIST & Tiny ImageNet & CelebA \\
\toprule

\multirow{8}{*}{Time (s)}
& \multirow{2}{*}{Setup} & {\fname} & 0.90 & 0.36 & 1.93 & 5.91 & 0.27 & 11.55 & 24.52 & 52.82 \\
\cline{3-11}
& & {\cite{liu2023iprivjoin}} & 0.02 & 0.01 & 0.03 & 0.20 & 0.01 & 0.17 & 0.38 & 0.97 \\
\cmidrule[0.7pt]{2-11}

& \multirow{3}{*}{Offline} & \multirow{2}{*}{{\fname}} & \textbf{28.35} & \textbf{923.67} & \textbf{1428.97} & \textbf{358.01} & \textbf{11571.50} & \textbf{3257.21} & \textbf{23697.00} & \textbf{52772.25} \\
& & & (2.15$\times$) & (2.05$\times$) & (2.11$\times$)& (2.27$\times$) & (2.02$\times$) & (2.32$\times$) & (0.23$\times$) & (2.44$\times$) \\
\cline{3-11}
& & {\cite{liu2023iprivjoin}} & 61.01 & 1894.27 & 3017.93 & 811.16 & 23396.34 & 7557.74 & 5559.89 & 128964.29 \\
\cmidrule[0.7pt]{2-11}

& \multirow{3}{*}{\textbf{Online}}
& \multirow{2}{*}{{\fname}}
& \cellcolor[HTML]{DAE8FC}\textbf{1.17} & \cellcolor[HTML]{DAE8FC}\textbf{1.46} & \cellcolor[HTML]{DAE8FC}\textbf{7.11} & \cellcolor[HTML]{DAE8FC}\textbf{8.24} & \cellcolor[HTML]{DAE8FC}\textbf{6.28} & \cellcolor[HTML]{DAE8FC}\textbf{41.15} & \cellcolor[HTML]{DAE8FC}\textbf{287.37} & \cellcolor[HTML]{DAE8FC}\textbf{694.13} \\
& & & \cellcolor[HTML]{DAE8FC}(2.54$\times$) & \cellcolor[HTML]{DAE8FC}(13.28$\times$) & \cellcolor[HTML]{DAE8FC}(15.20$\times$) & \cellcolor[HTML]{DAE8FC}(7.51$\times$) & \cellcolor[HTML]{DAE8FC}(19.54$\times$) & \cellcolor[HTML]{DAE8FC}(15.55$\times$) & \cellcolor[HTML]{DAE8FC}(19.34$\times$) & \cellcolor[HTML]{DAE8FC}(19.64$\times$) \\
\cline{3-11}
& & {\cite{liu2023iprivjoin}} & \cellcolor[HTML]{EFEFEF}2.97 & \cellcolor[HTML]{EFEFEF}19.37 & \cellcolor[HTML]{EFEFEF}108.07 & \cellcolor[HTML]{EFEFEF}61.88 & \cellcolor[HTML]{EFEFEF}122.78 & \cellcolor[HTML]{EFEFEF}639.93 & \cellcolor[HTML]{EFEFEF}5557.13 & \cellcolor[HTML]{EFEFEF}13631.98 \\

\toprule
\multirow{6}{*}{Comm (MB)}
& \multirow{3}{*}{Offline} & \multirow{2}{*}{{\fname}} & \textbf{40.85} & \textbf{690.57} & \textbf{3427.70} & \textbf{1810.01} & \textbf{5690.00} & \textbf{21968.86} & \textbf{189022.21} & \textbf{483810.00} \\
& & & (2.54$\times$) & (2.45$\times$) & (2.46$\times$) & (2.53$\times$) & (2.43$\times$) & (2.53$\times$) & (2.53$\times$) & (2.54$\times$) \\
\cline{3-11}
& & {\cite{liu2023iprivjoin}} & 103.89 & 1689.20 & 8438.50 & 4574.48 & 13831.38 & 55595.27 & 478611.54 & 1226790.31 \\
\cmidrule[0.7pt]{2-11}

& \multirow{3}{*}{\textbf{Online}}
& \multirow{2}{*}{{\fname}}
& \cellcolor[HTML]{DAE8FC}\textbf{1.21} & \cellcolor[HTML]{DAE8FC}\textbf{10.25} & \cellcolor[HTML]{DAE8FC}\textbf{60.87} & \cellcolor[HTML]{DAE8FC}\textbf{33.98} & \cellcolor[HTML]{DAE8FC}\textbf{68.76} & \cellcolor[HTML]{DAE8FC}\textbf{365.20} & \cellcolor[HTML]{DAE8FC}\textbf{3135.53} & \cellcolor[HTML]{DAE8FC}\textbf{7650.45} \\
& & & \cellcolor[HTML]{DAE8FC}(6.56$\times$) & \cellcolor[HTML]{DAE8FC}(6.58$\times$) & \cellcolor[HTML]{DAE8FC}(6.33$\times$) & \cellcolor[HTML]{DAE8FC}(7.29$\times$) & \cellcolor[HTML]{DAE8FC}(6.38$\times$) & \cellcolor[HTML]{DAE8FC}(6.31$\times$) & \cellcolor[HTML]{DAE8FC}(6.34$\times$) & \cellcolor[HTML]{DAE8FC}(6.34$\times$) \\
\cline{3-11}
& & {\cite{liu2023iprivjoin}} & \cellcolor[HTML]{EFEFEF}7.96 & \cellcolor[HTML]{EFEFEF}67.44 & \cellcolor[HTML]{EFEFEF}385.51 & \cellcolor[HTML]{EFEFEF}247.85 & \cellcolor[HTML]{EFEFEF}438.45 & \cellcolor[HTML]{EFEFEF}2304.43 & \cellcolor[HTML]{EFEFEF}19885.10 & \cellcolor[HTML]{EFEFEF}48527.58 \\

\toprule
\multicolumn{2}{c}{\multirow{3}{*}{Memory Cost (MB)}} & \multirow{2}{*}{{\fname}} & \textbf{19.82} & \textbf{83.23} & \textbf{437.18} & \textbf{181.54} & \textbf{493.84} & \textbf{2555.14} & \textbf{21937.85} & \textbf{53518.57} \\
\multicolumn{2}{c}{} & & (1.40$\times$) & (2.61$\times$) & (2.61$\times$) & (3.75$\times$) & (2.63$\times$) & (2.67$\times$) & (2.68$\times$) & (2.56$\times$) \\
\cline{3-11}
\multicolumn{2}{c}{} & {\cite{liu2023iprivjoin}} & 27.83 & 217.25 & 1142.05 & 680.80 & 1299.34 & 6830.83 & 58822.61 & 136947.61 \\

\bottomrule
\end{tabular}
}
\label{tab:reald}
\end{table*}
Initially, the parties collaboratively compute the secret-shared means, $\share{\mathit{X\_mean}}$ and $\share{\mathit{Y\_mean}}$, for each vector (Lines 1-2). Subsequently, they calculate the deviation of each element from its respective mean, producing the secret-shared difference vectors $\share{\mathit{X\_dif}}$ and $\share{\mathit{Y\_dif}}$ (Lines 3-6). The protocol then proceeds to compute the core components of the Pearson formula: the secret-shared covariance, which is the dot product of the two difference vectors (Line 7), and the sum of squared deviations for each vector, $\share{X\_dev}$ and $\share{Y\_dev}$ (Lines 8-9).
Finally, a crucial optimization in our proposed protocol is avoiding a high-overhead secure division. The standard Pearson formula involves a secure division by the product of standard deviations.
Secure division in secret sharing uses Newton-Raphson's method for finding a modular inverse, which requires $\mathcal{O}(log(k))$ rounds of multiplication where $k$ is the required precision. Instead, our protocol first computes the secret-shared inverse square roots of $\share{X\_dev}$ and $\share{Y\_dev}$ (Lines 10-11). This allows the final Pearson correlation coefficient $\share{v}$ to be efficiently computed using only secure multiplications (Line 12), thereby reducing communication rounds and overall overhead.

\subsection{Additional Evaluation}
\label{app:e}

\subsubsection{Comparison across Real-World Datasets}
\label{app:e_d}
As shown in Table~\ref{tab:reald}, we compare \fname with \baseline~\cite{liu2023iprivjoin} in three phases using eight real-world datasets. The experimental results have been analyzed in the main text and are not further discussed here.

\subsubsection{Comparison on Small-Scale Datasets under Varying Parameters}
\label{app:e_small}
We conduct evaluations for varying row count $n$ and varying feature dimension $m$ separately.

We conduct experiments on datasets of row count $n \in  \{2^{12},2^{16},$ $2^{20}\} $ with a fixed feature dimension (\(m_a = m_b = 100\)).
As shown in Table~\ref{tab:e2en}, \fname significantly outperforms the \baseline by \(4.17 \sim 10.10\times\) in running time and \(7.43 \sim 7.58\times\) in communication size.
When the dataset is small (e.g., $n = 2^{12}$), our improvements over the \baseline are relatively limited. This is primarily because, under such conditions, both the dataset size and feature dimensionality are low. In this regime, the ECC computation for IDs within our protocol becomes the dominant overhead, thereby creating an efficiency bottleneck.
Furthermore, the communication size reduction of \fname remains at about \(86.63\%\).

\begin{table}[htbp]
\belowrulesep=0pt
\aboverulesep=0pt
\centering
\caption{
Comparison of \fname vs. \baseline~\cite{liu2023iprivjoin} (feature dimensions \(m_a = m_b = 100\)) across dataset sizes $n$.
}
\setlength{\tabcolsep}{5pt}
\renewcommand{\arraystretch}{1.15}
\scalebox{0.88}{
\begin{tabular}{c|ccc|ccc}
\toprule
      \multirow{2}{*}{Protocol}    & \multicolumn{3}{c|}{Time (s)}                                                       & \multicolumn{3}{c}{Communication (MB)}                                             \\ \cline{2-7}
         & $2^{12}$ & $2^{16}$ & $2^{20}$ & $2^{12}$ & $2^{16}$ & $2^{20}$ \\ \toprule
{\small \fname}       &   3.20   & 19.29   & 306.12   & 6.65      & 106.40     & 1702.40     \\
{\small \baseline\cite{liu2023iprivjoin}}  & 13.37  & 194.30  & 3090.55  & 50.45     & 791.03     & 12653.15
\\ \bottomrule
\end{tabular}
}
\label{tab:e2en}
\end{table}

We vary the feature dimensions $m_a = m_b \in \{ 100, 500, 1000\}$ and fix the dataset row count $n = 2^{16}$.
As shown in Table \ref{tab:e2em}, \fname outperforms \baseline by $11.85\sim17.26\times$ in terms of running time and $7.43\sim 7.69\times$ in terms of communication size.
In addition, the performance improvement of \fname becomes more pronounced as the feature dimension $m$ increases.
For instance, as \(m_a = m_b\) increases from 100 to 1000, the running time improvement of \fname over \baseline increases from $11.85\times$ to $17.26\times$.

\begin{table}[htbp]
\belowrulesep=0pt
\aboverulesep=0pt
\centering
\caption{Comparison of \fname vs. \baseline \cite{liu2023iprivjoin} (dataset size \(n = 2^{16}\)) across varing feature dimension $m_a$ and $m_b$.
}
\setlength{\tabcolsep}{5pt}
\scalebox{0.88}{
\begin{tabular}{c|ccc|ccc}
\toprule
    \multirow{2}{*}{Protocol}        & \multicolumn{3}{c|}{Time (s) }     & \multicolumn{3}{c}{Communication (MB)} \\ \cline{2-7}
 & 100      & 500      & 1000       & 100        & 500         & 1000        \\ \toprule
{\small \fname}     & 19.34   & 53.24   & 96.26   & 106.40     & 506.40      & 1006.40    \\
{\small \baseline\cite{liu2023iprivjoin}} & 194.06  & 951.19  & 1895.38 & 791.03     & 3880.70     & 7742.78
\\ \bottomrule
\end{tabular}
}
\label{tab:e2em}
\end{table}

\subsubsection{Comparison in Offline Phase}
\label{app:e_off}
In the offline phase, both \fname and  \baseline~\cite{liu2023iprivjoin} only rely on $\mathcal{F}_{\textit{O-Shuffle}}$.
\baseline introduces a local extension method to expand the output of $\mathcal{F}_{\textit{O-Shuffle}}$ from a single dimension to $m$ dimensions. Although this local extension method improves offline efficiency, it introduces a 1-bit error. As a result, the extension creates an inaccurately joined table shares, which reduces the accuracy of subsequent SDA tasks. The experimental results on the Skin Cancer dataset in \baseline\cite{liu2023iprivjoin} show a decrease in accuracy due to this error, with accuracy dropping from $88.47\%$ to $88.11\%$, representing a decrease of approximately $0.41\%$.
Consequently, we use the oblivious shuffle protocol $\Pi_{\rm O\text{-}Shuffle}$ without the extension method to produce precise random values in the offline phase. To ensure a fair comparison, both our protocol and \baseline disable the local extension method and implement the same version of $\Pi_{\rm O\text{-}Shuffle}$ during evaluation.

\begin{table}[htbp]
\belowrulesep=0pt
\aboverulesep=0pt
\centering
\caption{Offline comparison of \fname vs. \baseline\cite{liu2023iprivjoin}.}
\setlength{\tabcolsep}{5pt}
\renewcommand{\arraystretch}{1.15}
\scalebox{0.88}{
\begin{tabular}{c|ccc|ccc}
\toprule
      \multirow{2}{*}{Protocol}    & \multicolumn{3}{c|}{Time (s)}                                                       & \multicolumn{3}{c}{Communication (MB)}                                             \\ \cline{2-7}
         & $2^{12}$ & $2^{16}$ & $2^{20}$ & $2^{12}$ & $2^{16}$ & $2^{20}$ \\ \toprule
{\small \fname}       &   27.62                     & 81.98                     & 1016.84                   & 36.31                     & 558.27                    & 10774.62                  \\
{\small \baseline\cite{liu2023iprivjoin}}  & 60.30                     & 200.45                    & 2663.15                   & 92.91                     & 1484.26                   & 28513.83
\\ \bottomrule
\end{tabular}
}
\label{tab:off}
\end{table}

We conduct offline experiments on datasets of varying row count \( n \in \{2^{12}, 2^{16}, 2^{20}\} \), with a fixed feature dimension (\( m_a = m_b = 10 \)). As shown in Table~\ref{tab:off}, \fname outperforms the \baseline by an average of $2.42\times$ in terms of time. This significant improvement is primarily due to \fname's enhanced communication complexity, which reduces the communication size by an average of $2.62\times$ compared to the \baseline.

\subsubsection{Comparison with unbalanced secure data join protocol~\cite{songsuda}.}
\label{app:e_unbalance}
One special scenario is the unbalanced setting~\cite{hao2024unbalanced,songsuda}, where the row count of the two parties' datasets differs by at least two orders of magnitude.
We conduct experiments with the SOTA secure unbalanced and redundancy-free data join protocol \suda~\cite{songsuda}.
Let $n_a$ and $n_b$ represent the input data row counts of parties $P_a$ and $P_b$, respectively.
We vary the feature dimensions $m_a = m_b \in \{ 100, 640\}$ and vary $P_a$'s row count $n_a \in \{1024, 2048\}$ when fixing $P_b$'s row count $n_b = 2^{20}$.
The results in Table~\ref{tab:unbalance} show that \fname achieves $2.24 \sim 5.04\times$ faster running time compared to \suda. However, this performance gain comes with an increase in communication size. This can be explained by the fact that \suda leverages homomorphic encryption to perform polynomial operations on encrypted ciphertexts. As a result, \suda incurs higher computational costs but benefits from reduced communication overhead.

\begin{table}[htbp]
\caption{Comparison of \fname vs. unbalanced protocol \suda~\cite{songsuda}.}
\belowrulesep=0pt
\aboverulesep=0pt
\scalebox{0.88}{
\setlength{\tabcolsep}{4pt}
\begin{tabular}{cc|cBcc|cc}
\toprule
\multirow{2}{*}{$n_b$} & \multirow{2}{*}{$n_a$}  & \multirow{2}{*}{Protocol} & \multicolumn{2}{c|}{$m_a = m_b = 100$} & \multicolumn{2}{c}{$m_a = m_b = 640$} \\ \cline{4-7}
                     &                       &                           & Time (s)        & Comm. (MB)        & Time (s)        & Comm. (MB)        \\ \toprule
\multirow{4}{*}{$2^{20}$}  & \multirow{2}{*}{1024} & {\small \fname}                      & 73.71           & 832.85            & 443.15          & 5157.07           \\
                     &                       & {\small \suda~\cite{songsuda}}                      & 171.32          & 154.42            & 993.57          & 508.01            \\ \cline{2-7}
                     & \multirow{2}{*}{2048} & {\small \fname}                      & 73.90           & 833.70            & 445.87          & 5162.14           \\
                     &                       &  {\small \suda~\cite{songsuda}}                      & 353.81          & 226.57            & 2245.70         & 933.71     \\
\bottomrule
\end{tabular}
}
\label{tab:unbalance}
\end{table}

\newcommand{\Prove}{\mathsf{Prove}}
\newcommand{\Verify}{\mathsf{Verify}}
\newcommand{\crs}{\mathsf{crs}}
\newcommand{\Com}{\mathsf{Com}}

\subsection{Strengthening the Threat Model}
\label{app:malicious}

\rthree{
In this section, we elaborate on the path to strengthen the threat model of \fname.
\fname comprises $\Pi_{\rm SMIG}$ (Protocol~\ref{pro:smig}) and $\Pi_{\rm MI\text{-}SFA}$ (Protocol~\ref{pro:sfa}), which are secure against semi-honest adversaries.
In contrast to semi-honest adversaries who follow all protocol specifications but passively attempt to learn additional information, malicious adversaries may try to deviate from the protocol to either obtain private information or compromise its functionality.
The main challenges in strengthening \fname to secure against malicious adversaries are twofold:
(i) ensuring that the ECC-based ID encrypting step in Protocol~\ref{pro:smig} is executed correctly by both parties, as this is the main interactive component;
(ii) ensuring that the oblivious shuffle protocol in Protocol~\ref{pro:sfa} is performed correctly, as this is the only interactive process in Protocol~\ref{pro:sfa}).}

\rthree{We can strengthen \fname against these malicious behaviors by (i) augmenting $\Pi_{\rm SMIG}$ (Protocol~\ref{pro:smig}) with a commitment scheme~\cite{pedersen1991non} and Non-Interactive Zero-Knowledge proofs (NIZKs)~\cite{fiat1986prove,groth2016size}, and (ii) replacing underlying oblivious shuffle protocol $\Pi_{\rm O\text{-}Shuffle}$ in $\Pi_{\rm MI\text{-}SFA}$ (Protocol~\ref{pro:sfa}) with malicious-secure oblivious shuffle protocol in Figure~12 of~\cite{song2023secret}. Since the latter is a modular, drop-in substitution, we focus below on augmenting $\Pi_{\rm SMIG}$ (Protocol~\ref{pro:smig}) using the commit-and-prove paradigm.}

\smallskip
\noindent\textit{\rthree{Zero-knowledge notation.}}
\rthree{Let $(\mathbb{G}, +)$ be a prime-order cyclic group of order $q$ generated by $g$.
Let $\Com(\cdot;\cdot)$ be a commitment scheme (e.g., instantiated with Pedersen commitments~\cite{pedersen1991non}) with standard binding and hiding properties.
Let $\Pi_{\mathsf{NIZK}}$ be a NIZK system.
For a relation $\mathcal{R}$, $x$ denotes the public statement, $w$ denotes the corresponding private witness such that $(x, w) \in \mathcal{R}$, and a proof $\pi^{zkp}$ allows a prover to convince a verifier of the validity of $x$ without leaking any information about $w$.}

\smallskip
\noindent \textit{\rthree{Commit phase.}}
\rthree{
First, we introduce a commitment phase before the Setup step in Protocol~\ref{pro:smig}. $P_a$ generates commitments $C_{\IDa} = \mathsf{Com}(\IDa; r_{\IDa})$, $C_{\pi^a_1} = \mathsf{Com}(\pi^a_1; r_{\pi^a_1}), C_{\pi^a_2} = \mathsf{Com}(\pi^a_2; r_{\pi^a_2})$, $C_{\alpha} = \mathsf{Com}(\alpha; r_{\alpha})$ (collectively abbreviated as $C_a =\Com(\IDa, \pi^a_1$, $\pi^a_2, \alpha;$ $ R_a)$) to its identifiers, its two permutations, and its private key, using randomness. Similarly, $P_b$ use randomness to generate commitments $C_{\IDb} = \mathsf{Com}(\IDb; r_{\IDb})$, $C_{\pi^b_1} = \mathsf{Com}(\pi^b_1; r_{\pi^b_1}), C_{\pi^b_2} = \mathsf{Com}(\pi^b_2; r_{\pi^b_2})$,$ C_{\beta} = \mathsf{Com}(\beta; r_{\beta})$ (collectively abbreviated as $C_b = \Com( $ $\IDb, \pi^b_1, \pi^b_2, \beta; R_b)$). Two parties exchange these commitments with communication complexity $\mathcal{O}(n)$ for $n$ identifiers. All subsequent NIZKs are proven relative to these commitments, and any verification failure causes an immediate protocol abort.}

\smallskip
\noindent \textit{\rthree{Adding prove-and-verify in three substeps $1.(1)$--$1.(3)$.}}
\rthree{
In Online Step 1-(1), when $P_a$ sends $\alpha H(\IDaa)$ to $P_b$, it appends a
proof $\pi^{zkp}_1$. This proof $\pi^{zkp}_1$ demonstrates knowledge of the values committed in $C_a =\Com(\IDa, \pi^a_1, \pi^a_2, \alpha; R_a)$ and that $\alpha H(\IDaa)$ is correctly computed such that $(\alpha H(\IDaa))_{\pi^a_1(i)} = \alpha H(\IDa_i)$ for all $i \in [n]$.
This proof ensures $P_a$ uses its committed identifiers $\IDa$, permutation $\pi^a_1$ and private key $\alpha$ to generate the shuffled and encrypted identifiers.
$P_b$ verifies $\pi^{zkp}_1$ using $C_a$ and the received $\alpha H(\IDaa)$, aborting if verification fails.}

\rthree{
In Step 1-(2), When $P_b$ sends $\beta\alpha H(\IDaab)$ and $\beta H(\IDbb)$, it appends a NIZK proof $\pi^{zkp}_2$. This proof demonstrates knowledge of $(\IDb, \pi^b_1, \pi^b_2, \beta; R_b)$ corresponding to $C_b$ such that: (1) $\beta H(\IDbb)$ is correctly computed with $(\beta H(\IDbb))_{\pi^b_2(j)} = \beta H(\IDb_j)$ for all $j \in [n]$ using the committed $\IDb$, $\pi^b_2$, and $\beta$; and (2) $\beta\alpha H(\IDaab)$ is correctly computed by applying the committed permutation $\pi^b_1$ and key $\beta$ to the received $\alpha H(\IDaa)$ such that $(\beta\alpha H(\IDaab))_{\pi^b_1(i)} = \beta(\alpha H(\IDaa))_i$ for all $i \in [n]$. $P_a$ verifies $\pi^{zkp}_2$ using $C_b$, the previously sent $\alpha H(\IDaa)$, and the received $\beta\alpha H(\IDaab)$ and $\beta H(\IDbb)$, aborting if verification fails.}

\rthree{
In Step 1-(3), after $P_a$ decrypts $\beta\alpha H(\IDaab)$ using its private key $\alpha$ to obtain $\beta H(\IDaab) = [\alpha^{-1}(x)]_{x \in \beta\alpha H(\IDaab)}$, it computes the mapped intersection index pairs $\MII$. When $P_a$ sends the output $\MII$ to $P_b$, it appends a NIZK proof $\pi^{zkp}_3$. This proof demonstrates that: (1) the decryption was performed correctly using the committed $\alpha$ from $C_a$ such that $(\beta H(\IDaab))_i = \alpha^{-1}((\beta\alpha H(\IDaab))_i)$ for all $i \in [n]$, and (2) $\MII$ is computed correctly from $\beta H(\IDaab)$ and $\beta H(\IDbb)$ relative to $P_a$'s commitment $C_a$, i.e., $\MII$ contains exactly the pairs $(i, \pi^a_2(j))$ for all $(i, j)$ where $(\beta H(\IDaab))_i = (\beta H(\IDbb))_j$, and no other pairs.
$P_b$ verifies $\pi^{zkp}_3$ using $C_a$, the previously sent $\beta\alpha H(\IDaab)$ and $\beta H(\IDbb)$, and the received $\MII$, accepting $\MII$ as the final output only if verification succeeds.
Across three steps 1.(1)–1.(3), all proofs incur $\mathcal{O}(n)$ communication overhead for $n$ identifiers.
}

\end{document}